\documentclass{elsarticle}
\usepackage{color,soul}
\setulcolor{none} 
\sethlcolor{white}  

\soulregister\cite7
\soulregister\ref7
\soulregister\pageref7

\usepackage[utf8]{inputenc}
\usepackage{amsmath,amsfonts}
\usepackage{amssymb,mathrsfs,amsthm}
\usepackage{graphicx}
\usepackage{subfig}
\usepackage{array}
\usepackage[export]{adjustbox} 
\usepackage{enumerate}

\DeclareMathOperator{\bin}{bin}
\DeclareMathOperator{\ind}{ind}
\DeclareMathOperator{\avg}{avg}

\DeclareMathOperator*{\dist}{dist}

\newtheorem{fact}{Fact}
\newtheorem{thm}{Theorem}
\newtheorem{proposition}{Proposition}
\newtheorem{lemma}{Lemma}
\newdefinition{example}{Example}

\usepackage{mathtools}

\DeclarePairedDelimiter\norm{\lVert}{\rVert}

\begin{document}

\begin{frontmatter}
\title{On the decomposition of stochastic cellular automata}  

\author[pan,gent]{Witold Bołt\corref{c1}}
\ead{witold.bolt@hope.art.pl}
\cortext[c1]{Corresponding author}

\author[gent]{Jan M.~Baetens}
\author[gent]{Bernard De~Baets}

\address[pan]{Systems Research Institute, Polish Academy of Sciences,\\Newelska St.\ 6, 01-447~Warsaw,~Poland}
\address[gent]{KERMIT, Department of Mathematical Modelling, Statistics and Bioinformatics,\\ Ghent University, Coupure links 653, B-9000 Gent, Belgium}

\begin{abstract}
In this paper we present two interesting properties of stochastic cellular automata that can be helpful in analyzing the dynamical behavior of such automata. The first property allows for calculating cell-wise probability distributions over the state set of a stochastic cellular automaton, {\it i.e.}\ images that show the average state of each cell during the evolution of the stochastic cellular automaton. The second property shows that stochastic cellular automata are equivalent to so-called stochastic mixtures of deterministic cellular automata. Based on this property, any stochastic cellular automaton can be decomposed into a set of deterministic cellular automata, each of which contributes to the behavior of the stochastic cellular automaton. 
\end{abstract}

\begin{keyword}
stochastic cellular automata \sep complexity analysis \sep continuous cellular automata \sep decomposition
\end{keyword}
\end{frontmatter}

\section{Introduction}
\label{sec:intro}
Cellular Automata (CAs) are often used for constructing models in a variety of fields of application, including chemistry, biology, medicine, physics, ecology and the study of socioeconomic interactions. In many of these settings, Stochastic CAs (SCAs) are considered due to the stochastic nature of the phenomenon under study or due to a lack of understanding of the exact rules driving the phenomenon \cite{Ganguly03asurvey,das2012survey}. Therefore, a better understanding of the dynamics of SCAs is crucial. Only a few methods for dealing with the analysis of models based on SCAs have been developed. In many practical applications, especially in cases where the averaged behavior of the system is of concern, sampling methods, relying on extensive computer simulations, are sufficient \cite{aids,2136045,4167259}. Techniques built on the mean--field theory can be used to study the long-term behavior of SCAs \cite{RSA:RSA20126}. The theory of Markov chains can be applied to provide analytical tools for analyzing the model's behavior~\cite{Mairesse201442}, although in practice, due to the theoretical and computational complexity of such tools, the application scope is limited. 

This paper is devoted to providing effective analytical tools based on deterministic CAs for the analysis of multi-state SCAs. Although the theoretical foundations of the presented results are already (at least partially) known in the literature \cite{Mairesse201442,lebowitz,computing-pca}, so far the applications are limited. Therefore, the main aim of this paper is to provide a complete, formal description of the discussed properties in a form that is suitable for applications and that does not require a strong mathematical background, as well as to present examples that can motivate further applications in the domain of systems modeling. The methods presented here are developed in the context of 1D SCAs on finite lattices, but can be easily generalized to the case of higher dimensions and infinite spaces.

Two main results are presented in this paper. The first one involves constructing images that show the cell-wise probability distribution over the state set, at any time step. The method is based on associating an SCA with a deterministic, Continuous CA (CCA). The second result shows the equivalence of SCAs and stochastic mixtures of deterministic CAs. Based on this finding, any SCA can be decomposed into a finite set of deterministic CAs, each of them contributing to the behavior of the stochastic system. An effective method for finding a decomposition is presented. It allows to uncover the deterministic component in the mixture with the highest impact on the behavior of the SCA.

This paper is organized as follows. We start with some preliminaries and definitions in Section~\ref{sec:preliminary}. In Section~\ref{sec:cca}, we introduce the concept of CCAs and the formalism enabling the analysis of multi-state CAs. Section \ref{sec:multi-sca} contains the definition of multi-state SCAs and holds the main results of this paper. The paper is concluded with Section \ref{sec:experiments}, discussing the experimental results that illustrate our results. A summary is presented in Section~\ref{sec:summary}.

\section{Preliminaries}
\label{sec:preliminary}
Informally, a CA is a discrete dynamical system in which the space is subdivided into discrete elements, referred to as cells. At every consecutive, discrete time step, each cell is assigned one of $N$ states using a deterministic rule, which depends only on the previous state of the considered cell and the states of its neighboring cells \cite{neumann1948}.

Formally speaking, let the state set $\mathcal{S}$ be a finite set of $N>1$ elements. Elements of the set $\mathcal{C} = \{ c_i \mid i=1,\dotsc,M \}$ will be referred to as cells. Every cell $c_i$ is assigned a state $s(c_i,t)\in \mathcal{S}$ at each time step $t\in\mathbb{N}$, according to a local, deterministic rule. The vector $s(\cdot, t) \in \mathcal{S}^M$ will be referred to as the configuration, and as the initial configuration when $t=0$. The sequence $(s(\cdot,0),s(\cdot,1),\ldots)$ will be referred to as the space-time diagram of the CA. For technical reasons, we impose periodic boundary conditions, but our results do not depend on this assumption.

The function $A\colon \mathcal{S}^M \to \mathcal{S}^M$ satisfying, for every $t\in\mathbb{N}$:
\begin{equation}
s(\cdot, t+1) = A(s(\cdot, t))\,, 
\end{equation}
will be referred to as a global CA rule or simply a CA, if there exists a radius $r\in\mathbb{N}$ and a function $f\colon\mathcal{S}^{2\,r+1}\to\mathcal{S}$ satisfying:
\begin{equation}
s(c_i, t+1) = f(s(c_{i-r}, t), \dotsc, s(c_{i+r},t))\,,
\label{eq:local-fun}
\end{equation}
for every $i$ and $t\in\mathbb{N}$. Such a function $f$ will be referred to as a local rule. Note that a local rule uniquely defines the global rule, while for a given global rule, multiple local rules exist. Additionally, it is assumed that $r\ll M$. The vector $(c_{i-r}, c_{i-r+1},\dotsc,c_{i+r-1},c_{i+r})$ will be referred to as the neighborhood of cell $c_i$ and $R=2\,r+1$ will denote the neighborhood size. For the sake of simplicity, $s(c_{i-r},\dotsc,c_{i+r}, t)\in\mathcal{S}^R$ will denote the state of the neighborhood of cell $c_i$ at time step $t$, and will be referred to as the neighborhood configuration of $c_i$ at time step $t$.

\section{Continuous CAs}
\label{sec:cca}

There exist multiple ways of extending the definition of CAs to cover infinite state sets. Examples of such approaches include Coupled Map Lattices (CMLs)~\cite{kaneko93} and so-called fuzzy CAs \cite{flocchini2000convergence,Betel:2009:ABF:1621142.1621283,Betel20095}. In this section, we present Continuous CAs (CCAs), which can be seen as a generalization of the ideas presented in \cite{flocchini2000convergence}. Our formalism is based on a polynomial representation of discrete CA rules. We start with formulating the continuous counterparts of binary CAs. After that we present a generalization to cover multi-state CAs. 

\subsection{Binary CAs}
\label{sec:binary-cca}

Binary CAs are widely studied \cite{RevModPhys.55.601,wolfram-class}, because they allow to evolve complex patterns and exhibit complex behavior despite their intrinsic simplicity. The state set of such a CA $A$ is $\mathcal{S} = \{0,1\}$. We will now formally define and characterize its local rule $f\colon \mathcal{S}^R\to\mathcal{S}$. Let $l = (l_i)_{i=1}^{2^R}$ be a binary vector. We consider a system of equations:
\begin{equation}
f(s_{i,1},\dotsc,s_{i,R}) = l_i\,,
\end{equation}
where $(s_{i,1},\dotsc,s_{i,R})$ is a binary vector such that $i = 1 + \sum_{j=0}^{R-1} s_{i,R-j}\,2^{j}$. As can be seen, such a system of equations is uniquely defined by the vector $l$. The vector $l$ will be referred to as the lookup table (LUT) of the local rule $f$. It is not difficult to check that such a system uniquely defines the function $f$, since it lists all of the possible input configurations, and maps them to corresponding values by components of the vector $l$.

Following \cite{flocchini2000convergence}, we know that the function $f$ can be expressed as a polynomial, which is of interest for our purposes. In order to define it, we introduce two auxiliary functions. We start with the function $\ind\colon \{1,\dotsc,2^R\}\to \{1,2\}^R$. It is defined in such a way that $\ind(i)[m]$ is the $m$--th digit, incremented by one, read from left to right, of the binary representation of the integer $i-1$, padded with ones on the left, so that it always has length $R$. Consequently, it holds that:
\begin{equation}
i = 1+\sum_{m=1}^{R} (\ind(i)[R-(m-1)]-1)\,2^{m-1}.
\label{eq:bin-def}
\end{equation}
The values of $\ind(i)$ for $R=3$ and $i\in\{1,\dotsc,8\}$ are shown below:
\begin{gather*}
\ind(1) = (1,1,1),\ \ind(2) = (1,1,2),\ \ind(3) = (1,2,1),\ \ind(4) = (1,2,2),\\
\ind(5) = (2,1,1),\ \ind(6) = (2,1,2),\ \ind(7) = (2,2,1),\ \ind(8) = (2,2,2).
\end{gather*}
The function $\ind$ is related to the binary representation of integers. In \cite{flocchini2000convergence} a simpler formulation using the function $\bin$, which yields the binary representation of an integer, is used. The construction presented here, although a bit more complicated in the binary setting, allows for a smoother generalization to the multi-state case. 

Using the function $\ind$, we now define the function $n\colon \mathcal{S} \times \mathbb{N} \times \mathbb{N} \to \mathcal{S}$, which for $s\in\mathcal{S}$ and $m,i\in\mathbb{N}$, is given by:
\[
n(s, m, i) =
\begin{cases}
s &, \textrm{if}\ \ind(i)[m] = 2,\\
1- s &, \textrm{if}\ \ind(i)[m] = 1.
\end{cases}
\]
Note that we will use vectors of states of the form $(s_1,\dotsc,s_R)\in\mathcal{S}^R$ and for simplicity, for any $m\in\{1,\dotsc,R\}$, we will write $n(s_m,i)$ instead of $n(s_m, m, i)$. Using the functions $\ind$ and $n$, we can write the polynomial representation of the local rule $f$ as:
\begin{equation}
f(s_1,\ldots,s_{R}) = \sum_{i=1}^{2^R} l_i \, \left(\prod_{m=1}^{R} n(s_m,i) \right)\,.
\label{eq:fca-gen}
\end{equation}
The following example shows the explicit form of this polynomial for a member of the family of Elementary CAs (ECAs).

\begin{example}[Elementary CAs]
\label{exa:one}
Binary, 1D CAs with neighborhood radius $r=1$ are commonly referred to as ECAs~\cite{RevModPhys.55.601}. There are 256 such ECAs. Treating the LUT entries $l_i$ as digits of an integer written in base 2, we can enumerate the local rules of ECAs. By convention, the binary vectors are read in the reverse order, {\it i.e.}\ $(l_i)_{i=1}^8$ is interpreted as $(l_8, l_7, \dotsc, l_1)_2$. For example, given the LUT $l = (0,1,1,0,1,0,0,1)$ of ECA 150, and denoting the Boolean complement as $\bar{s} = 1-s$, its local rule can be written, according to Eq.~$(\ref{eq:fca-gen})$, as:
\[ 
f_{150}(s_1,s_2,s_3) = \bar{s}_1\,\bar{s}_2\,s_3
+ \bar{s}_1\,s_2\,\bar{s}_3 
+ s_1\,\bar{s}_2\,\bar{s}_3
+ s_1\,s_2\,s_3\,.\tag*{\qed}\]
\end{example}

Using the above notation, a CCA can be defined analogously to a binary CA, with two notable differences. Firstly, the state set of a CCA is the unit interval, {\it i.e.}\ $\mathcal{S}=[0,1]$, and, secondly, the local rule $f\colon [0,1]^R\to[0,1]$ is given by Eq.~(\ref{eq:fca-gen}) with coefficients $l_i\in [0,1]$. We will refer to such a vector $(l_i)_{i=1}^{2^R}$ as a generalized~LUT. It is easy to check that this definition of a CCA is formally correct. Indeed, the values of the function $f$ in Eq.\ (\ref{eq:fca-gen}) are guaranteed to belong to the unit interval if $l_i\in[0,1]$ for all $i$ and $s_m\in[0,1]$ for all $m~\in~\{1,\dotsc,R\}$. Note that this construction is directly related to the one presented in~\cite{flocchini2000convergence}, where fuzzy CAs are constructed as polynomials representing fuzzified logical functions. Following the same line of reasoning, an alternative polynomial representation for the local rules of binary CAs is presented in \cite{schule2008global}, which is consistent with a logical representation of the local rules.  

In order to introduce the formalism that is needed in the multi-state setting, we present a slightly modified way of representing binary CAs compared to the one obtained through Eq.~$(\ref{eq:fca-gen})$. Let us assume that the state set is given by $\mathcal{S}_2 = \{ (1,0), (0,1) \}\subset \mathbb{R}^2$. Then the local rule is a function $f\colon \mathcal{S}_2^R\to \mathcal{S}_2$ and can be represented as a vector function $f = (f_1, f_2)$, where $f_j\colon \mathcal{S}_2^R\to \{0,1\}$. Note that any $s = (s_1, s_2)\in \mathcal{S}_2$ satisfies $s_2 = 1-s_1$. Similarly, $f_2 = 1 - f_1$, so that the local rule $f$ can be defined using $f_2$ only. \hl{Let $l\in\mathcal{S}_2^{2^R}$ be given, {\it i.e.} for every $i=1,\dotsc,2^R$, it holds that $l_i = (l_{i,1}, l_{i,2})\in \mathcal{S}_2$. Let $(s_1,\dotsc,s_R) \in \mathcal{S}_2^R$.} We can define $f_2$ with a formula similar to Eq.\ (\ref{eq:fca-gen}) as:
\begin{equation}
f_2(s_1,\dotsc,s_{R}) = \sum_{i=1}^{2^R} l_{i,2} \, \left(\prod_{m=1}^{R} n(s_{m,2},i) \right).
\end{equation}
Since $s_{m,1} = 1 - s_{m,2}$, we can rewrite the formula for the function $n$ as:
\[
n(s_{m,2}, i) =
\begin{cases}
s_{m,2} &, \textrm{if}\ \ind(i)[m] = 2,\\
s_{m,1} &, \textrm{if}\ \ind(i)[m] = 1,
\end{cases}
\]
which can be simplified to:
\begin{equation}
n(s_{m,2}, i) = s_{m,\ind(i)[m]}\,,
\end{equation}
and thus $f_2$ can be written as:
\begin{equation}
f_2(s_1,\ldots,s_{R}) = \sum_{i=1}^{2^R} l_{i,2} \, \left(\prod_{m=1}^{R} s_{m,\ind(i)[m]} \right).
\end{equation}

Having defined the function $f_2$, we can prove the following fact, which gives a direct formula for $f_1$. 
\begin{lemma}
The function $f_1$ can be expressed as:
\begin{equation}
f_1(s_1,\ldots,s_{R}) = \sum_{i=1}^{2^R} l_{i,1} \, \left(\prod_{m=1}^{R} s_{m,\ind(i)[m]} \right).
\end{equation}
\end{lemma}

\noindent Due to the above, we can write a general expression for $f_j$, $j=1,2$, as:
\begin{equation}
f_j(s_1,\ldots,s_{R}) = \sum_{i=1}^{2^R} l_{i,j} \, \left(\prod_{m=1}^{R} s_{m,\ind(i)[m]} \right).
\label{eq:final-cca-2}
\end{equation}
Resorting to this vector notation, a CCA can be defined with a state set: 
\begin{equation}
S_c = \{(x_1, x_2) \in [0,1]^2 \mid x_1+x_2 = 1\}\subset \mathbb{R}^2\,,
\label{eq:sc-2}\end{equation}
and a local rule $f\colon S_c^R \to S_c$ given by Eq.\ (\ref{eq:final-cca-2}) with $(l_j)_{j=1}^{2^R} \in S_c^{2^R}$.

\begin{example}
The local rule $f_{150}$ of ECA 150 presented in Example \ref{exa:one} can be rewritten in the context of the state set $\mathcal{S}_2$, as $f\colon \mathcal{S}^3_2 \to \mathcal{S}_2$, $f=(f_1, f_2)$. Note that for $s_i\in\mathcal{S}_2$, it holds that $s_i=(s_{i,1}, s_{i,2})$. The local rule $f$ can be written as: 
\begin{align*}
f_1(s_1, s_2, s_3) & = s_{1,2}\,s_{2,2}\,s_{3,1} + s_{1,2}\,s_{2,1}\,s_{3,2} + s_{1,1}\,s_{2,2}\,s_{3,2} + s_{1,1}\,s_{2,1}\,s_{3,1} \\
f_2(s_1, s_2, s_3) & = s_{1,2}\,s_{2,2}\,s_{3,2} + s_{1,2}\,s_{2,1}\,s_{3,1} + s_{1,1}\,s_{2,2}\,s_{3,1} + s_{1,1}\,s_{2,1}\,s_{3,2}\,.\\[-\normalbaselineskip]\tag*{\qed}
\end{align*}
\end{example}

The representation given by Eq.~(\ref{eq:final-cca-2}) is equivalent to that given by Eq.~(\ref{eq:fca-gen}), but allows for an easier generalization to multi-state CAs. Therefore, we will use it throughout the remainder of this paper.

\subsection{Multi-state CAs}
\label{sec:multi-ca}
We start with generalizing the state set $\mathcal{S}_2$ defined in Section \ref{sec:binary-cca} to the case of $N$ states. Let $\mathcal{S}_N\subset \mathbb{R}^N$ be the set of all base vectors of the $N$--dimensional Euclidean space, {\it i.e.}
\begin{equation}
\mathcal{S}_N =\Big\{ e \in \{0,1\}^N \mid \norm{e} = 1\Big\}\,.
\label{eq:state-set}
\end{equation}
Let the $i$--th state in a multi-state CA be represented by the base vector $e_i$ which has zeros everywhere, except at position $i$, where it has a one. For instance, if $N=3$, then $\mathcal{S}_3 = \{ (1,0,0), (0,1,0), (0,0,1)\}$. For the sake of simplicity, in the remainder of this paper, we will write $\mathcal{S}$ instead of $\mathcal{S}_N$, assuming that $N$ is fixed.

We now generalize the function $\ind \colon \{1,\dotsc,N^R\} \to \{1,\dotsc,N\}^R$, defined above for the binary case, to enumerate the $N^R$ neighborhood configurations. Let $i\in\{1,\dotsc,N^R\}$ be the index identifying a neighborhood configuration. We assume that such a configuration consists of the following states:
\begin{equation}
(e_{\ind(i)[1]}, e_{\ind(i)[2]},\dotsc,e_{\ind(i)[R]})\,.
\end{equation}
We assume that $\ind$ satisfies the following equation:
\begin{equation}
i = 1 + \sum_{m=1}^{R}\Big(\ind(i)[R-(m-1)]-1\Big)\,N^{m-1}\,.
\end{equation}
A direct formula for calculating $\ind$ is given by:
\begin{equation}
\ind(i)[m] = \left\lfloor\frac{(i-1)\ \textrm{mod}\ N^{R-(m-1)}}{N^{R-m}}\right\rfloor + 1\,.
\end{equation}

Essentially, $\ind(i)$ yields the positional representation of $i-1$ in base $N$, padded with zeros on the left, so that it always has length $R$, where each of the digits is incremented by one. Hence, one easily finds:
\begin{align*}
\ind(1) &= (1,\dotsc,1,1)\,, \\
\ind(2) &= (1,\dotsc,1,2)\,, \\
\vdots \\
\ind(N) &= (1,\dotsc,1,N)\,, \\
\ind(N+1) &= (1,\dotsc,2,1)\,, \\
\vdots\\
\ind(N^R) &= (N,\dotsc,N,N)\,.
\end{align*}

Any local rule $f$ of an $N$--state, deterministic CA can be uniquely defined by writing down its outputs for all of the possible neighborhood configurations. Formally, let us consider the system of equations:
\begin{equation} 
f(e_{\ind(i)[1]},\dotsc,e_{\ind(i)[R]}) = l_{i},
\end{equation}
where $i\in\{1,\dotsc,N^R\}$ and $l_i \in \mathcal{S}$. Similarly to the binary case, the matrix $L= (l_i)_{i=1}^{N^R}\in\mathcal{S}^{N^R}$ will be called the LUT of a multi-state CA. Basically, $L$ is a matrix of $N^{R}$ columns and $N$ rows, containing only zeros and ones, in such a way that every column contains exactly one non-zero entry. Each of the matrices $L\in\mathcal{S}^{N^{R}}$ uniquely defines an $N$--state CA in terms of its local rule.

In order to define a CCA in the context of multiple states, we need to formally define its state set. Let $\mathcal{S}_c$ be defined as:
\begin{equation}
\mathcal{S}_c = \left\{ (x_1, \dotsc, x_N) \in [0,1]^N \mid \sum_{i=1}^{N} x_i = 1\right\}\subset \mathbb{R}^N\,,
\end{equation}
which is consistent with the definition for $N=2$ given by Eq.~$(\ref{eq:sc-2})$. Note that any $s\in\mathcal{S}_c$ satisfies $s_i = 1 - \sum_{j\neq i} s_j$. The set $\mathcal{S}_c$ is commonly referred to as a standard $(N-1)$--simplex \cite{rudin-principles}. 

\begin{example}
Let $N=3$. The set $\mathcal{S}_c$ is a 2D triangle placed in the 3D Euclidean space, with vertices $(1,0,0)$, $(0,1,0)$ and $(0,0,1)$. \qed
\end{example}

In the case of a CCA, the local rule is a vector function $f\colon \mathcal{S}_c^{R}\to\mathcal{S}_c$, given by $f=(f_1,\dotsc,f_N)$, for which there exists a matrix $P\in\mathcal{S}_c^{N^R}$, referred to as the generalized LUT, such that for $j\in\{1,\dotsc,N\}$ it holds that $f_j\colon\mathcal{S}_c^R\to[0,1]$ and:
\begin{equation}
f_j(s_1,\dotsc,s_R) = \sum_{i=1}^{N^R} P_{ij}\left(\prod_{m=1}^{R}s_{m,\ind(i)[m]}\right)\,.
\label{eq:def-f}
\end{equation}

We will show that the definition given by Eq.\ $(\ref{eq:def-f})$ is formally correct, meaning that the function $f$ with components $f_j$ satisfies $f\colon \mathcal{S}_c^{R}\to\mathcal{S}_c$, which is equivalent to showing that $f_j(s_1,\dotsc,s_R) \in [0,1]$ and $\sum_{j=1}^N f_j(s_1,\dotsc,s_R) = 1$ for any $(s_1,\dotsc,s_R)\in \mathcal{S}_c^R$.  

\begin{lemma}
For any $(s_1,\dotsc,s_R)\in\mathcal{S}_c^R$, it holds that
\begin{equation}
\sum_{i=1}^{N^R}\left[\prod_{m=1}^R s_{m,\ind(i)[m]}\right] = 1\,.
\label{eq:lem1}
\end{equation}
\end{lemma}

\begin{proof}
Let
\[\Theta = \sum_{i=1}^{N^R}\left[\prod_{m=1}^R s_{m,\ind(i)[m]}\right].\]
Note that for every $j\in\{1,\dotsc,N\}^R$, there exists exactly one $i\in\{1,\dotsc,N^R\}$ such that $j = \ind(i)$. Due to this, it holds that:
\[ \Theta = \sum_{j\in\{1,\dotsc,N\}^R}\left[\prod_{m=1}^{R} s_{m,j_m}\right].\]
By regrouping the elements of the above sum, we may write:
\[ \Theta = \sum_{j\in\{1,\dotsc,N\}^{R-1}}\left[\prod_{m=1}^{R-1} s_{m,j_m}\left(\sum_{i=1}^{N} s_{R,i}\right)\right].\]
Since $s_R \in \mathcal{S}_c$, we know that $\sum_{i=1}^{N} s_{R,i} = 1$, and thus:
\[ \Theta = \sum_{j\in\{1,\dotsc,N\}^{R-1}}\left[\prod_{m=1}^{R-1} s_{m,j_m}\right].\]
By repeating this regrouping procedure $R-1$ times, we finally get:
\[ \Theta = \sum_{j=1}^N s_{1,j} = 1\,.\qedhere\]
\end{proof}

\begin{proposition}
For any $j\in\{1,\dotsc,N\}$ and any $(s_1,\dotsc,s_R)\in\mathcal{S}_c^R$, it holds that
$f_j(s_1,\dotsc,s_R) \in [0,1]$.
\end{proposition}
\begin{proof}
Firstly, $f_j(s_1,\dotsc,s_R)\geq 0$ is fulfilled since all of the elements in the sum in Eq.\ (\ref{eq:def-f}) are non-negative. Therefore, showing that $f_j(s_1,\dotsc,s_R)\leq 1$ is sufficient to prove the proposition. Since $P_{ij}\leq 1$, we can write:
\[ f_j(s_1,\dotsc,s_R) = \sum_{i=1}^{N^R}P_{ij}\left(\prod_{m=1}^R s_{m,\ind(i)[m]}\right)\leq
\sum_{i=1}^{N^R}\prod_{m=1}^R s_{m,\ind(i)[m]} = 1\,. \]
\end{proof}

\begin{proposition}
For any $(s_1,\dotsc,s_R)\in\mathcal{S}_c^R$, it holds that:
\[\sum_{j=1}^N f_j(s_1,\dotsc,s_R) = 1\,.\]
\end{proposition}
\begin{proof}
Let $j\in\{1,\dotsc,N\}$, then it suffices to show that:
\begin{equation}
1 - f_j(s_1,\dotsc,s_R) = \sum_{k\neq j} f_k(s_1,\dotsc,s_R)\,.
\end{equation}
This identity is proven by the following chain of equalities:
\begin{align*}
1 - f_j(s_1,\dotsc,s_R) &= 1 - \sum_{i=1}^{N^R}P_{ij}\left(\prod_{m=1}^R s_{m,\ind(i)[m]}\right) \\
&=1 - \sum_{i=1}^{N^R}\left(1-\sum_{k\neq j}P_{ik}\right)\prod_{m=1}^R s_{m,\ind(i)[m]}  \\
&=1 - \sum_{i=1}^{N^R}\prod_{m=1}^R s_{m,\ind(i)[m]} + \sum_{i=1}^{N^R}\sum_{k\neq j}P_{ik}\left(\prod_{m=1}^R s_{m,\ind(i)[m]}\right) \\
&= \sum_{k\neq j}\sum_{i=1}^{N^R}P_{ik}\left(\prod_{m=1}^R s_{m,\ind(i)[m]}\right)  
= \sum_{k\neq j} f_k(s_1,\dotsc,s_R)\,.\qedhere
\end{align*}
\end{proof}

We have shown the correctness of the CCA definition. We conclude this section with the observation that the local rule $f\colon \mathcal{S}^R\to\mathcal{S}$ of any $N$--state, deterministic, 1D CA can always be written in the form of Eq.~$(\ref{eq:def-f})$, where $(s_1,\dotsc,s_R)\in\mathcal{S}^R$, $f_j\colon \mathcal{S}^R\to\{0,1\}$ and $\sum_{j=1}^{N} f_j(s_1,\dotsc,s_R) = 1$. This shows that CCAs are a generalization of $N$--state, deterministic CAs. 

\section{Properties of SCAs}
\label{sec:multi-sca}
In this section, we provide a general definition of an $N$--state SCA, which is formulated using the notation introduced in Section \ref{sec:cca} for multi-state CAs and CCAs. Subsequently, we show the nature of the relation between SCAs and CCAs. Finally, we present a method for decomposing an SCA into a set of deterministic CAs.

\subsection{General construction of multi-state SCAs}
We define SCAs, the stochastic counterpart of multi-state CAs. We assume that the state set $\mathcal{S}$ is given by Eq.\ (\ref{eq:state-set}) and contains all the standard base vectors $e_i$. In such a setting, the states are assigned according to a probability distribution, which depends on the neighborhood configuration at the previous time step. Therefore, SCAs are not defined through a local rule, but rather by a set of conditional expressions, \hl{which for $j \in\{1,\dotsc,N\}$ and $k \in \{ 1,\dotsc,N^R\}$ can be written as}:
\begin{equation} 
\mathbb{P}\Big(s(c_i, t+1) = e_j \mid s(c_{i-r},\dotsc,c_{i+r},t) = (e_{\ind(k)[1]},\dotsc,e_{\ind(k)[R]})\Big)  = p_{kj}\,.
\label{eq:pkj}
\end{equation}
\hl{The conditional probabilities $p_{kj}$ in Eq.~(\ref{eq:pkj}) do not depend on the choice of $i$ and $t$. Obviously, it holds that $p_k = (p_{k1}, \dotsc, p_{kN}) \in \mathcal{S}_c$ for all $k$. Furthermore, the matrix $P = (p_{kj})$ can be considered an element of $\mathcal{S}_c^{N^R}$. Consequently, every element in $\mathcal{S}_c^{N^{R}}$}
uniquely defines the probability distribution of some local rule. The matrix $P$ will be referred to as the probabilistic lookup table (pLUT). Such matrices are commonly known as stochastic matrices \cite{stoma} and are typically used in the analysis of Markov chains.

Any SCA is uniquely defined by a pLUT belonging to $\mathcal{S}_c^{N^{R}}$. Since we know from Section \ref{sec:multi-ca} that every element of $\mathcal{S}_c^{N^{R}}$ uniquely defines a CCA, we elaborate on the relationship between an SCA defined by a pLUT $P\in\mathcal{S}_c^{N^R}$ and a CCA defined by the same $P$ considered as generalized LUT. 

\subsection{Relationship between CCAs and SCAs}
In contrast to deterministic CAs, only the initial state $s(\cdot, 0)$ is known with certainty in the case of SCAs, while for $s(\cdot, t)$, $t>0$, only a probability distribution can be calculated. Let $\pi(c_i, t\mid I_0)$ denote the probability distribution over the set of states $\mathcal{S}$ for the $i$--th cell at time step $t$, given that the evolution started from the initial condition $I_0$, {\it i.e.}\ $s(\cdot, 0) = I_0$. Whenever it will not bring confusion, we will omit $I_0$ and write $\pi(c_i, t)$. Since there are $N$ states, we need to specify $N$ probabilities summing up to 1, to define $\pi(c_i, t)$. For this reason, let us represent $\pi(c_i, t)$ as a vector of probabilities, {\it i.e.}\ $\pi(c_i, t)\in\mathcal{S}_c$, such that:
\begin{equation}
\pi(c_i, t \mid I_0)[k]=  \mathbb{P}\Big(s(c_i, t) = e_j \mid s(\cdot, 0) = I_0\Big),
\end{equation}
where $j=1,\dotsc,N$. 

When using SCAs in practical problem settings, one might be interested in computing the values of $\pi(c_i, t)$ for any $t$. Normally, this is achieved by means of sampling methods \cite{aids,2136045,4167259}, {\it i.e.}\ multiple space-time diagrams are generated from the same initial condition and probability distributions are estimated, based on the frequencies of reaching any of the states in a given cell. Unfortunately, such an approach has serious drawbacks, especially in terms of the computational burden it brings along. Yet, one can calculate $\pi(c_i, t)$ much faster by relying on CCAs.

\begin{proposition}
Let $P\in\mathcal{S}_c^{N^R}$ be the pLUT of an SCA, and let $f\colon\mathcal{S}_c^{R} \to \mathcal{S}_c$ be a function defined according to Eq.\ $(\ref{eq:def-f})$ using $P$ as generalized LUT, then:
\begin{equation}
\pi(c_i, t+1) = f(\pi(c_{i-r}, t), \dotsc, \pi(c_{i+r}, t))\,,
\end{equation}
for any $i$ and $t$.
\label{prop:cca-mstate}
\end{proposition}

\begin{proof}
Let $k\in\{1,\dotsc,N^R\}$ and $j\in\{1,\dotsc,N\}$, then $p_{kj}$ is the probability of reaching state $e_j$ if the neighborhood configuration was known to be $(e_{\ind(k)[1]},\dotsc,e_{\ind(k)[R]})$, which will be referred to as the $k$--th neighborhood. Assuming that $\pi(\cdot, t)$ is known, the probability of cells $(c_{i-r},\dotsc,c_{i+r})$ being in the $k$--th neighborhood can be calculated with the following formula:
\begin{multline}
\mathbb{P}\Big(s(c_{i-r},\dotsc,c_{i+r}, t)=(e_{\ind(k)[1]},\dotsc,e_{\ind(k)[R]})\Big)\\ = \prod_{m=1}^{R} \pi(c_{i-r+m-1}, t)\big[\ind(k)[m]\big].
\label{eq:nei-prob}
\end{multline}
Due to the Total Probability Theorem, the probability of reaching state $e_j$ at time step $t+1$ is expressed as:
\begin{multline}
\mathbb{P}\Big(s(c_i, t+1) = e_j\Big) = \\
\sum_{k=1}^{N^R} \mathbb{P}\Big(s(c_i, t+1) = e_j \mid s(c_{i-r},\dotsc, c_{i+r}, t)=(e_{\ind(k)[1]},\dotsc,e_{\ind(k)[R]})\Big) \times\\
\mathbb{P}\Big(s(c_{i-r},\dotsc, c_{i+r}, t)=(e_{\ind(k)[1]},\dotsc,e_{\ind(k)[R]})\Big)\,.
\label{eq:total-prob}
\end{multline}
Substituting Eqs.~$(\ref{eq:pkj})$ and $(\ref{eq:nei-prob})$ in Eq.\ $(\ref{eq:total-prob})$, we get:
\begin{equation}
\mathbb{P}\Big(s(c_i, t+1) = e_j\Big) =  \sum_{k=1}^{N^R} 
p_{kj}\left(\prod_{m=1}^{R} \pi(c_{i-r+m-1}, t)\big[\ind(k)[m]\big]\right).
\label{eq:pre-f}
\end{equation}
Taking into account the definition of the function $f$ in Eq.\ $(\ref{eq:def-f})$, we can rewrite Eq.\ $(\ref{eq:pre-f})$ as:
\[
\mathbb{P}\Big(s(c_i, t+1) = e_j\Big) =  f_j(\pi(c_{i-r},t),\dotsc,\pi(c_{i+r},t))\,.
\]
Since $\pi(c_i, t+1)[j] = \mathbb{P}\big(s(c_i, t+1) = e_j\big)$, the following is proven:
\[\pi(c_i, t+1) = f(\pi(c_{i-r},t),\dotsc,\pi(c_{i+r},t))\,.\qedhere\]
\end{proof}

As shown above, the construction of CCAs is directly connected to the Total Probability Theorem, and thus it enables us to better understand the evolution of cell-wise probability distributions of SCAs. It is important to highlight that the evolution of those distributions is deterministic in SCAs, although the dynamical system itself is stochastic. The finding above agrees with \cite{Betel20095}, where a similar result in the case of binary CAs was shown. Yet, this paper was confined to study the asymptotic behavior of deterministic CAs. \hl{Interestingly, in~\cite{busic2013} an SCA is defined as a discrete-time deterministic dynamical system acting on the set of probability measures on the set of all configurations. Although the definition presented there is much more abstract, and suits mostly theoretical applications, it reassembles the main ideas of CCAs.}

\subsection{Decomposition of SCAs}
In the previous subsection, we have shown that CCAs are directly related to SCAs, since the evolution of a CCA is equivalent to the evolution of the cell-wise probability distributions of an SCA. Here, we uncover the relationship between SCAs and deterministic, $N$--state CAs, and discuss the possibility of decomposing an SCA into deterministic CAs. Let us start with an introductory example motivating our general construction.

\begin{example}[$\alpha$--asynchronous CAs]
Classically, states in deterministic CAs are updated synchronously, {\it i.e.}\ a new state is assigned to all cells simultaneously at every time step according to the local rule. Yet, different approaches of breaking the synchronicity of CAs have been proposed \cite{schonfisch1999synchronous}. Interestingly, the choice of the update scheme, which defines the order or timing of cell state updates, has very important repercussions on the dynamical properties of CAs~\cite{Baetens2012383}. Here we focus on so-called $\alpha$-asynchronous CAs ($\alpha$-ACAs)~\cite{Fates05anexperimental}.

Any $\alpha$-ACA is defined by its deterministic counterpart $A$ and a probability~$\alpha$, called the synchrony rate, which controls whether or not its cells are updated. More precisely, $\alpha$ is the probability of applying the local rule $f$ of $A$, {\it i.e.}:
\[
s(c_i, t+1) = 
\begin{cases} 
f(s(c_{i-r},\dotsc,c_{i+r},t)) &, \textrm{with probability $\alpha$}, \\
s(c_i, t) &, \textrm{otherwise}.
\end{cases}
\]
Note that if $\alpha=0$, such a system remains at its initial configuration, whereas the system is equivalent to CA $A$ if $\alpha=1$.\qed 
\end{example}

The essential property of the construction presented above is that an $\alpha$--ACA can also be seen as a synchronous, stochastic system in which the local rule of $A$ is selected with probability $\alpha$, while the identity rule is selected with probability~$1-\alpha$. Hence, we may say that the rule of a CA $A$ is stochastically mixed with the identity rule. 

The idea behind $\alpha$--ACA can be easily extended to the mixing of $q\geq 2$ deterministic rules. Consider synchronous, deterministic CAs $A_1, \dotsc, A_q$ and probabilities $\alpha_1, \dotsc, \alpha_q$ that satisfy $\sum_{i=1}^{q} \alpha_i = 1$. Then, we define an SCA $\widetilde{A}$ in which the state of every cell at every consecutive time step is obtained by evaluating one of the local rules, corresponding to one of the $A_i$. The rule is chosen randomly and independently for every cell at each time step with the selection probability for CA $A_i$ being $\alpha_i$.  Such systems, referred to as stochastic mixtures of deterministic CAs in the remainder of this paper, hold very interesting properties that are currently also investigated by others \cite{IOPORT.05887733}. For technical reasons, we assume that the local rules of $A_1, \dotsc, A_q$ are defined with a common radius $r\geq 0$. Since each local rule with radius $r'<r$ can be represented equivalently as a local rule with radius $r$, we are not loosing generality with this assumption. 

The first observation is a direct consequence of the definition presented above.
\begin{fact}
Assume that $r\geq 0$ is the radius of the neighborhood of the local rules of automata $A_1, \dotsc, A_q$. Let $R=2\,r+1$ and suppose that $L_i \in \mathcal{S}^{N^R}$ is the LUT of CA $A_i$, for $i=1,\dotsc,q$. Then the pLUT $P\in\mathcal{S}_c^{N^R}$ of the stochastic mixture~$\widetilde{A}$ satisfies $P = \sum_{i=1}^{q} \alpha_i L_i$.
\label{fac:seven}
\end{fact}

Sums of the form $\sum_{i=1}^{q} \alpha_i L_i$, where $\alpha_i \in [0,1]$ and $\sum_{i=1}^{q} \alpha_i=1$ are commonly referred to as convex combinations, which allows us to rephrase Fact \ref{fac:seven} as: the pLUT of a stochastic mixture of deterministic CAs is the corresponding convex combination of the LUTs of the mixed CAs.

Having defined the concept of a stochastic mixture, we are interested in characterizing this class of SCAs. Interestingly, we are able to show that in general there are no characteristics that distinguish stochastic mixtures from other SCAs, since any SCA can be represented as a stochastic mixture. This fact is not surprising if we envisage a stochastic mixture of deterministic CAs as a convex combination of LUTs of deterministic CAs. We can rely on the classical Krein--Milman theorem from convex set theory \cite{Krein1940}. It states that a compact convex set is the convex hull of its extreme points. In our context, the set $\mathcal{S}^{N^R}$ containing the LUTs of deterministic CAs is the set of extreme points of $\mathcal{S}_c^{N^R}$, which is convex and compact. The convex hull of $\mathcal{S}^{N^R}$ is the set of all possible convex combinations of elements from $\mathcal{S}^{N^R}$, which is precisely the set of all stochastic mixtures of deterministic CAs. Due to this, we can state the following theorem.

\begin{thm}
Any SCA can be represented as a stochastic mixture of a finite number of deterministic CAs.
\end{thm}

The decomposition of a stochastic mixture into a convex combination described above is not unique, therefore many methods for constructing the decomposition can be formulated. In the proof of Proposition \ref{prop:stoch-mix-decom}, we introduce one possible algorithm of which it will be shown in Proposition \ref{prop:motiv} that it uncovers the most influential component of the stochastic mixture. 

\begin{proposition}
Let $P\in\mathcal{S}_c^{N^R}$. Then there exists an integer $n^\ast$ such that for $i=1,2,\dotsc,n^\ast$, we can define coefficients $\alpha_i\in [0,1]$, $\sum_{i=1}^{n^\ast} \alpha_i = 1$ 
and matrices $L^i \in \mathcal{S}^{N^R}$, that satisfy:
\begin{equation}
P = \sum_{i=1}^{n^\ast} \alpha_i\,L^i\,.
\end{equation}
\label{prop:stoch-mix-decom}
\end{proposition}

\begin{proof}
We start by defining auxiliary matrices $P^m$, for $m\geq 0$. Let:
\begin{equation}
P^m = \begin{cases}
P &, \textrm{if }m=0, \\
P^{m-1} - \alpha_m\,L^m &, \text{if }m>0,
\end{cases} 
\label{eq:pm-mat}
\end{equation}
where, for $m>0$:
\begin{equation}
\alpha_m = \min_i\max_j P^{m-1}_{ij}\,,\label{eq:alpha-def}
\end{equation}
and $L^m = (L^m_{ij})$ such that:
\begin{equation}
L^m_{ij} =
\begin{cases} 
1, & \textrm{for }j = \min \Big\{j \mid P^{m-1}_{il} \leq P^{m-1}_{ij}, \textrm{for any $l\in\{1,\dotsc,N\}$}\Big\}\,, \\
0, & \textrm{otherwise}.
\end{cases}
\label{eq:lm-mat}
\end{equation}
Note that, for any $m$ and $i$, it holds that $L^m_{ij} = 1$ for a single value of $j$. Consequently $L^m \in \mathcal{S}^{N^R}$.

Firstly, note that $\alpha_1$ is one of the elements of $P^0=P$, so $\alpha_1\in [0,1]$. Additionally, it holds that $\alpha_1>0$, since if $\alpha_1=0$, then for some $i$, we would have $\max_j P_{ij}=0$, which contradicts the fact that $\sum_j P_{ij} = 1$ for any~$i$.

Secondly, note that for any $i$ and $j$, it holds that $P^0_{ij}\geq 0$. We will show next that the assumption $P^{m-1}_{ij}\geq 0$ leads to $P^{m}_{ij}\geq 0$. Indeed, we know that $P^m_{ij} \in \{ P^{m-1}_{ij}, P^{m-1}_{ij} - \alpha_m \}$. If $P^{m}_{ij} = P^{m-1}_{ij}$, then obviously $P^{m}_{ij}\geq 0$, while if $P^{m}_{ij} = P^{m-1}_{ij}-\alpha_m$, it holds that $P^{m-1}_{ij}\geq P^{m-1}_{il}$ for any $l$, meaning that $P^{m-1}_{ij} = \max_l P^{m-1}_{il}$. On the other hand, since $\alpha_m = \min_i\max_l P^{m-1}_{il}$, it holds that $P^{m-1}_{ij}\geq \alpha_m$. Since it was assumed that $P^{m-1}_{ij}\geq 0$, we see that $P^m_{ij}\geq 0$ as well. 

Thirdly, following the argument above, we can easily show that for every $i$, $j$ and $m$, it holds that $P^{m}_{ij} \leq P^{m-1}_{ij}$. This means that for every $i$, it holds that $\max_j P^{m}_{ij} \leq \max_j P^{m-1}_{ij}$, and thus $\alpha_{m+1} \leq \alpha_{m}$. Since $\alpha_1\in\, ]0,1]$, and $\alpha_m\geq 0$ for any $m$, we see that $\alpha_m\in[0,1]$.

Note that for any $i$, it holds that $\sum_j P^m_{ij} = \left(\sum_j P^{m-1}_{ij}\right) - \alpha_m$. Expanding this recursively, yields $\sum_j P^m_{ij} = (\sum_j P^{0}_{ij}) - \sum_{l=1}^m\alpha_l$. From the definition of a stochastic matrix, we know that $\sum_j P^{0}_{ij} =1$, such that $\sum_j P^m_{ij} = 1 - \sum_{l=1}^m\alpha_l$, which means that the column sums in matrix $P^m$ are equal. Since for all $m$ we have shown that $P^m_{ij}\geq 0$, we know that for any $m$, it holds that $1 - \sum_{l=1}^{m}\alpha_l \geq 0$, therefore $\sum_{l=1}^{m}\alpha_l \leq 1$. 

We further note that $\alpha_m = 0$, for some $m$, if and only if $P^{m-1}$ contains one zero column. Since we have shown that the column sums in $P^m$ are equal for any $m$, we know that $\alpha_m=0$ if and only if $P^{m-1} = \mathbf{0}$. Therefore, $\alpha_m > 0$ if and only if $P^{m-1}\neq \mathbf{0}$, which implies that in each of the columns of $P^{m-1}$, there is at least one non-zero entry.

Let $z(m) = \#\{(i,j) \mid P^m_{ij} = 0\}$ denote the number of zeros in matrix $P^m$. Note that if $P^m\neq \mathbf{0}$, then $z(m) < z(m+1)$, which follows from the fact that if $P^m\neq \mathbf{0}$, there exists $P^m_{ij}=\alpha_{m+1}>0$, and thus for at least one couple $(i,j)$, we know that $P^{m+1}_{ij} = P^m_{ij} - \alpha_{m+1} = 0$. Since $z(m)$ cannot be greater than the total number of elements in $P^m$, it is not possible that $z(m) < z(m+1)$ for all~$m$. As such, we know that there exist $n^\ast$ for which $P^{n^\ast} = \mathbf{0}$ and $P^{n^\ast -1}\neq \mathbf{0}$, and thus $\sum_{l=1}^{n^\ast}\alpha_l = 1$. Using the definition of $P^{n^\ast}$ multiple times, we get:
\[ \mathbf{0} = P^{n^\ast} = P^{n^\ast -1} - \alpha_{n^\ast}\,L^{n^\ast} = \ldots = P^0 - \alpha_1\,L^1 - \ldots - \alpha_{n^\ast}\,L^{n\ast} = P - \sum_{i=1}^{n^\ast}\alpha_i\,L^i,\]
such that $P = \sum_{i=1}^{n^\ast}\alpha_i\,L^i$.
\end{proof}

To illustrate the construction presented in the proof of Proposition \ref{prop:stoch-mix-decom}, every step of the algorithm is visualized in the following example.

\begin{example}
Although Proposition \ref{prop:stoch-mix-decom} deals with stochastic matrices with dimensions corresponding to pLUTs, it is obvious that the same method can be used to decompose a stochastic matrix of arbitrary size. For simplicity, we apply the method to an exemplary $4\times 3$ stochastic matrix $P$ given by:
\[
P = \begin{bmatrix}
0.6 & 0 & 0.2 & 0.4 \\
0.3 & 1 & 0.3 & 0.4 \\
0.1 & 0 & 0.5 & 0.2 
\end{bmatrix}.
\]
According to Eqs.~(\ref{eq:pm-mat}) and (\ref{eq:lm-mat}), we build matrices $P^0$ and $L^1$:
\[
P^0 = P = \begin{bmatrix}
{\bf 0.6} & 0 & 0.2 & {\bf 0.4} \\
0.3 & {\bf 1} & 0.3 & 0.4 \\
0.1 & 0 & {\bf 0.5} & 0.2 
\end{bmatrix},
\quad
L^1 = \begin{bmatrix}
1 & 0 & 0 & 1 \\
0 & 1 & 0 & 0 \\
0 & 0 & 1 & 0 
\end{bmatrix}.
\]
Note that we highlighted the selected maximal elements in each of the columns using boldface font. The boldface entries in $P^{0}$ correspond to the positions of the ones in matrix $L^1$. Note that in the last column of $P^0$, there is no unique maximal element. Yet, as given by Eq.\ (\ref{eq:lm-mat}), we pick the first of the entries from top to bottom. Picking the other possibility would affect the subsequent matrices $L^m$ and result in a different decomposition, but the coefficients $\alpha_m$ would not change. 

Following Eq.\ (\ref{eq:alpha-def}), we find that $\alpha_1 = 0.4$. We proceed and calculate $P^1$ and~$L^2$:
\[
P^1 = P^0 - \alpha_1\,L^1 = \begin{bmatrix}
0.2 & 0 & 0.2 & 0 \\
{\bf 0.3} & {\bf 0.6} & {\bf 0.3} & {\bf 0.4} \\
0.1 & 0 & 0.1 & 0.2 
\end{bmatrix},
\quad
L^2 = \begin{bmatrix}
0 & 0 & 0 & 0 \\
1 & 1 & 1 & 1 \\
0 & 0 & 0 & 0 
\end{bmatrix},
\]
and from this we see that $\alpha_2 = 0.3$. We continue the procedure:
\[
P^2 = P^1 - \alpha_2\,L^2 = \begin{bmatrix}
{\bf 0.2} & 0 & {\bf 0.2} & 0 \\
0 & {\bf 0.3} & 0 & 0.1 \\
0.1 & 0 & 0.1 & {\bf 0.2} 
\end{bmatrix},
\quad
L^3 = \begin{bmatrix}
1 & 0 & 1 & 0 \\
0 & 1 & 0 & 0 \\
0 & 0 & 0 & 1 
\end{bmatrix},
\]
and $\alpha_3 = 0.2$. We proceed one more step and get:
\[
P^3 = P^2 - \alpha_3\,L^3 = \begin{bmatrix}
0 & 0 & 0 & 0 \\
0 & {\bf 0.1} & 0 & {\bf 0.1} \\
{\bf 0.1} & 0 & {\bf 0.1} & 0 
\end{bmatrix},
\quad
L^4 = \begin{bmatrix}
0 & 0 & 0 & 0 \\
0 & 1 & 0 & 1 \\
1 & 0 & 1 & 0 
\end{bmatrix},
\]
and $\alpha_4 = 0.1$. This is the final step, since: $P^4 =  P^3 - \alpha_4\,L^4 = {\bf 0}$. Therefore, the decomposition can be written as:
\begin{align*}
P = 0.4 &\times \begin{bmatrix}
1 & 0 & 0 & 1 \\
0 & 1 & 0 & 0 \\
0 & 0 & 1 & 0 
\end{bmatrix} + 0.3 \times \begin{bmatrix}
0 & 0 & 0 & 0 \\
1 & 1 & 1 & 1 \\
0 & 0 & 0 & 0 
\end{bmatrix} \\ + 0.2 &\times \begin{bmatrix}
1 & 0 & 1 & 0 \\
0 & 1 & 0 & 0 \\
0 & 0 & 0 & 1 
\end{bmatrix} + 0.1 \times \begin{bmatrix}
0 & 0 & 0 & 0 \\
0 & 1 & 0 & 1 \\
1 & 0 & 1 & 0 
\end{bmatrix}.\\[-\normalbaselineskip]\tag*{\qed}
\end{align*}
\end{example}

As mentioned earlier, the construction presented in the proof of Proposition~\ref{prop:stoch-mix-decom} is one of many possibilities of decomposing a pLUT, but in the following proposition we will show that it enables us to capture the element of the mixture with the highest possible coefficient. 

\begin{proposition}
Let $P\in\mathcal{S}_c^{N^R}$, $p=1,\dotsc, p_{max}$, $q=1,\dotsc, q_{max}$, and let $\alpha_p, \beta_q \in [0,1]$, $L^p, M^q \in \mathcal{S}^{N^R}$ be such that $\sum_p \alpha_p = \sum_q \beta_q = 1$. Moreover, let $P=\sum_p \alpha_p\,L^p  = \sum_q \beta_q\,M^q$, and let $\alpha_p$, $L^p$ be defined as in the proof of Proposition \ref{prop:stoch-mix-decom}.
It then holds that $\max_q \beta_q \leq \alpha_1 = \max_p \alpha_p$.
\label{prop:motiv}
\end{proposition}

\begin{proof}
Let $i$ and $j$ be such that $P_{ij} = \alpha_1$. Then for any $k$, it holds that $P_{ik}\leq \alpha_1$. Let $q\in\{1,\dotsc,q_{max}\}$ and let $l$ be such that $M^q_{il} = 1$. Since $P~=~ \sum_q \beta_q\,M^q$, we know that $\beta_q \leq P_{il}$, and thus $\beta_q \leq P_{il} \leq \alpha_1$. Therefore, for any $q~\in~\{1,\dotsc,q_{max}\}$, it holds that $\beta_q \leq \alpha_1$ and thus $\max_q \beta_q \leq \alpha_1$.
\end{proof}

The informal meaning of the above proposition is that the presented approach to decompose an SCA uncovers the deterministic rule that has the highest probability of being executed, and therefore is likely to have the highest influence on the behavior of the system. The exact relations between the dynamical behavior of SCAs and the dynamical characteristics of the components of stochastic mixtures are still under investigation. 

\hl{The experiments in the following section underline that in some cases those CAs that have the highest probability of application indeed greatly influence the behavior of the SCA. Yet, in some cases, those with a very small probability of application can also play an important role. Therefore, at this stage we do not claim that the overall dynamics of the SCA is always determined by the component of the mixture with the highest probability of application.

The decomposition of an SCA as a stochastic mixture might also find its use in SCA-based modeling. Indeed, in some cases it might be easier to express the model in terms of deterministic components having a meaning in the language relevant to the modeling task. Due to the presented equivalence, we are guaranteed that such a practice will not limit the modeling potential.}    

\section{Experiments}
\label{sec:experiments}

\subsection{Analysis of $\alpha$--asynchronous ECAs}
\label{sec:aaca}

Through the following experiment, we will analyze the behavior of $\alpha$-asyn\-chro\-nous ECAs. For a given rule, the cell-wise distance between the space-time diagram of the deterministic CA and the space--time diagram of a CCA representation of an $\alpha$--ACA variant, for $\alpha$ ranging from $0$ to $1$, was measured. More formally, if $A$ is an ECA rule, and $A_\alpha$ is the CCA representation of the $\alpha$--asynchronous variant of $A$ for $\alpha\in[0,1]$, we measured the distance $D(\alpha)$ between space-time diagrams.
For a random initial condition $I_0\in \{0,1\}^M$, where $M>0$ denotes the number of cells and $T$ the number of time steps, the distance $D$ is defined as:
\begin{equation}
D(\alpha) = \frac{1}{M\,T} \sum_{t=1}^{T} \norm{A^t(I_0) - A^t_\alpha(I_0)}\,,
\end{equation}
where $A^t(I_0)$ denotes the result of applying the global rule $A$ $t$ times to input~$I_0$. In this experiment, we choose $M=T=69$.

Unsurprisingly, it turned out that the most pronounced discrepancies between the parallel evolutions are observed when $\alpha$ approaches 1. For that reason, we restrict the further discussion to $\alpha\in[0.9,1]$. We classified the ECAs according to the behavior of the function $D$ (the contents of the classes is shown in Table~\ref{tab:ex1-classes}):
\begin{itemize}
\item Class I: $D(\alpha)$ is almost 0, for all $\alpha\in[0.9,1]$ (Figure~\ref{fig:c1}),
\item Class II: $D(\alpha)$ smoothly decreases towards 0, as $\alpha$ increases (Figure~\ref{fig:c2}),
\item Class IIIa: there is a very sudden drop of $D(\alpha)$ as $\alpha$ approaches 1 (Figure~\ref{fig:c3a}),
\item Class IIIb: there is a very sudden drop of $D(\alpha)$ as $\alpha$ approaches 1, and the behavior is not monotonic (Figure~\ref{fig:c3b}). 
\end{itemize}

\begin{figure}
\centering
\subfloat[Class I (ECA 40)]{\label{fig:c1}\includegraphics[width=0.4\textwidth]{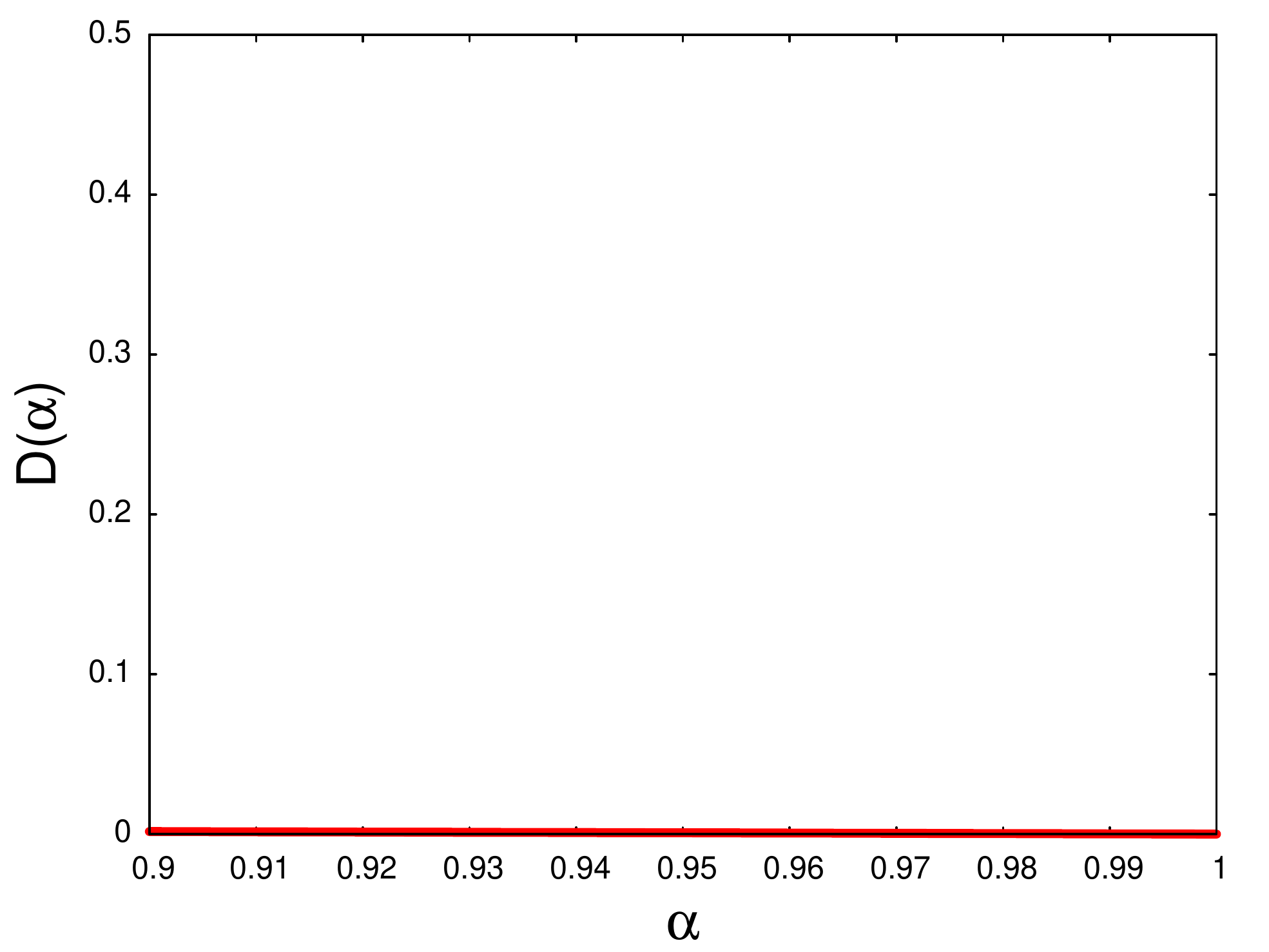}}\ 
\subfloat[Class II (ECA 42)]{\label{fig:c2}\includegraphics[width=0.4\textwidth]{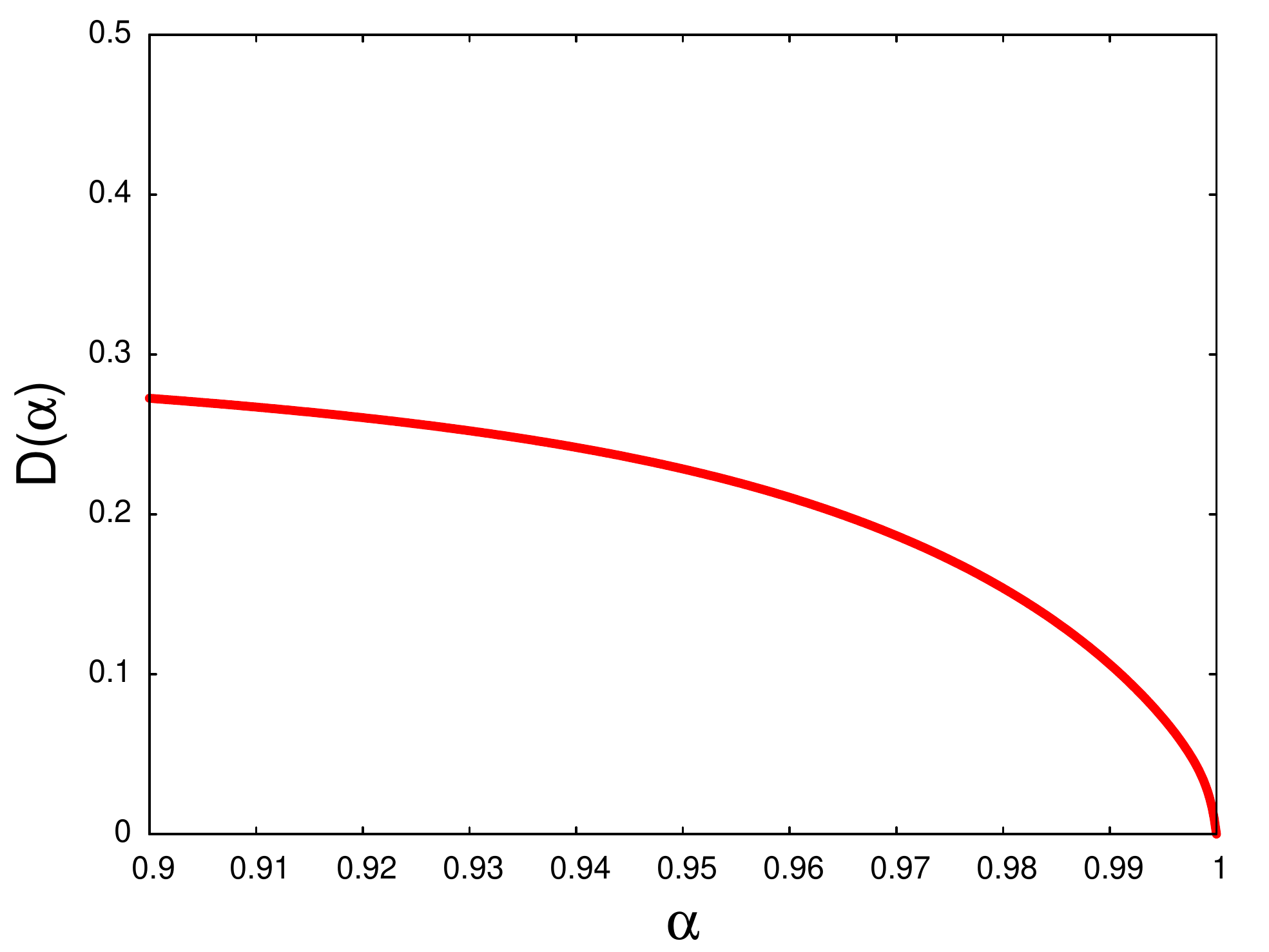}}\ 
\subfloat[Class IIIa (ECA 38)]{\label{fig:c3a}\includegraphics[width=0.4\textwidth]{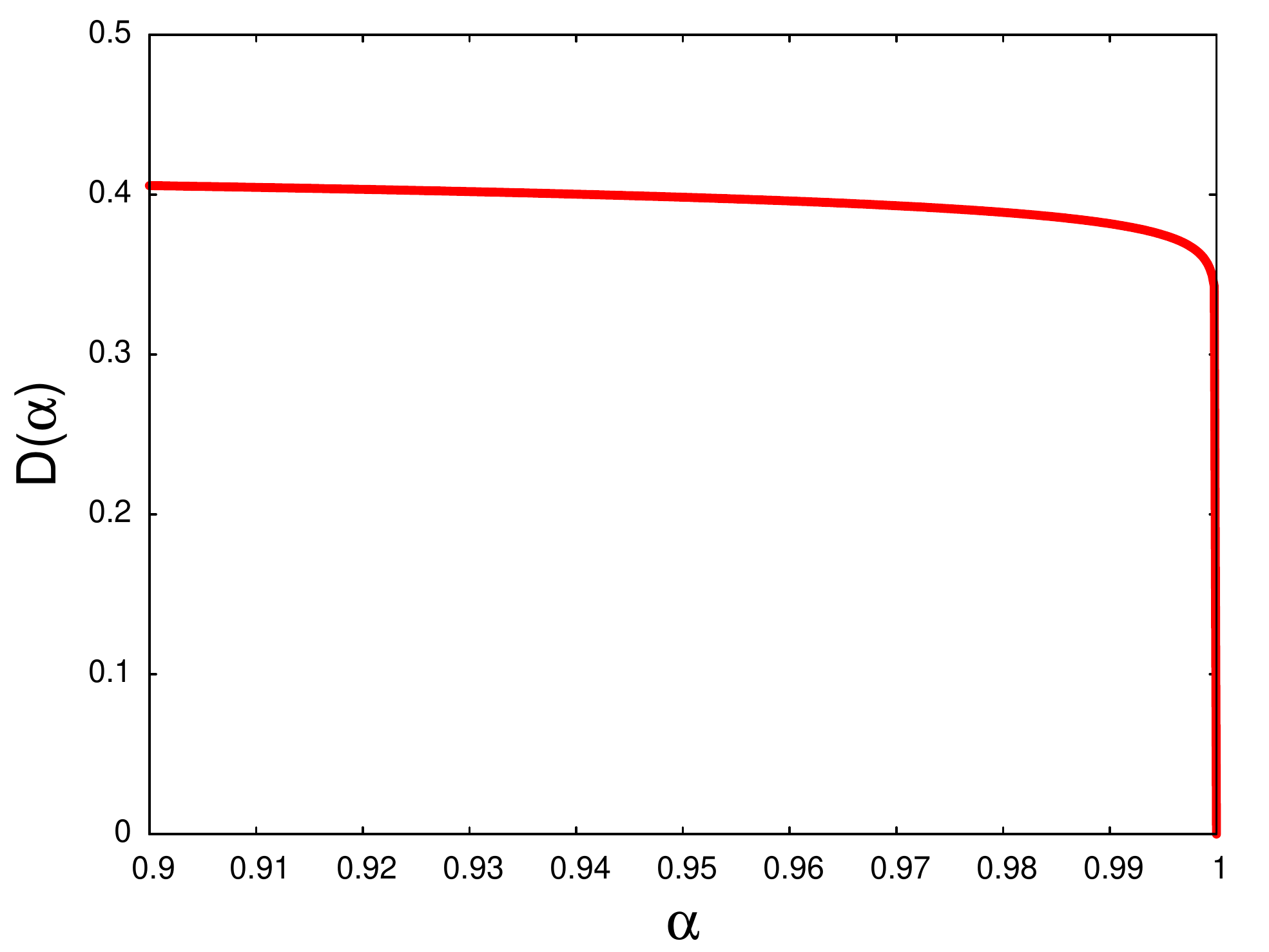}}\ 
\subfloat[Class IIIb (ECA 43)]{\label{fig:c3b}\includegraphics[width=0.4\textwidth]{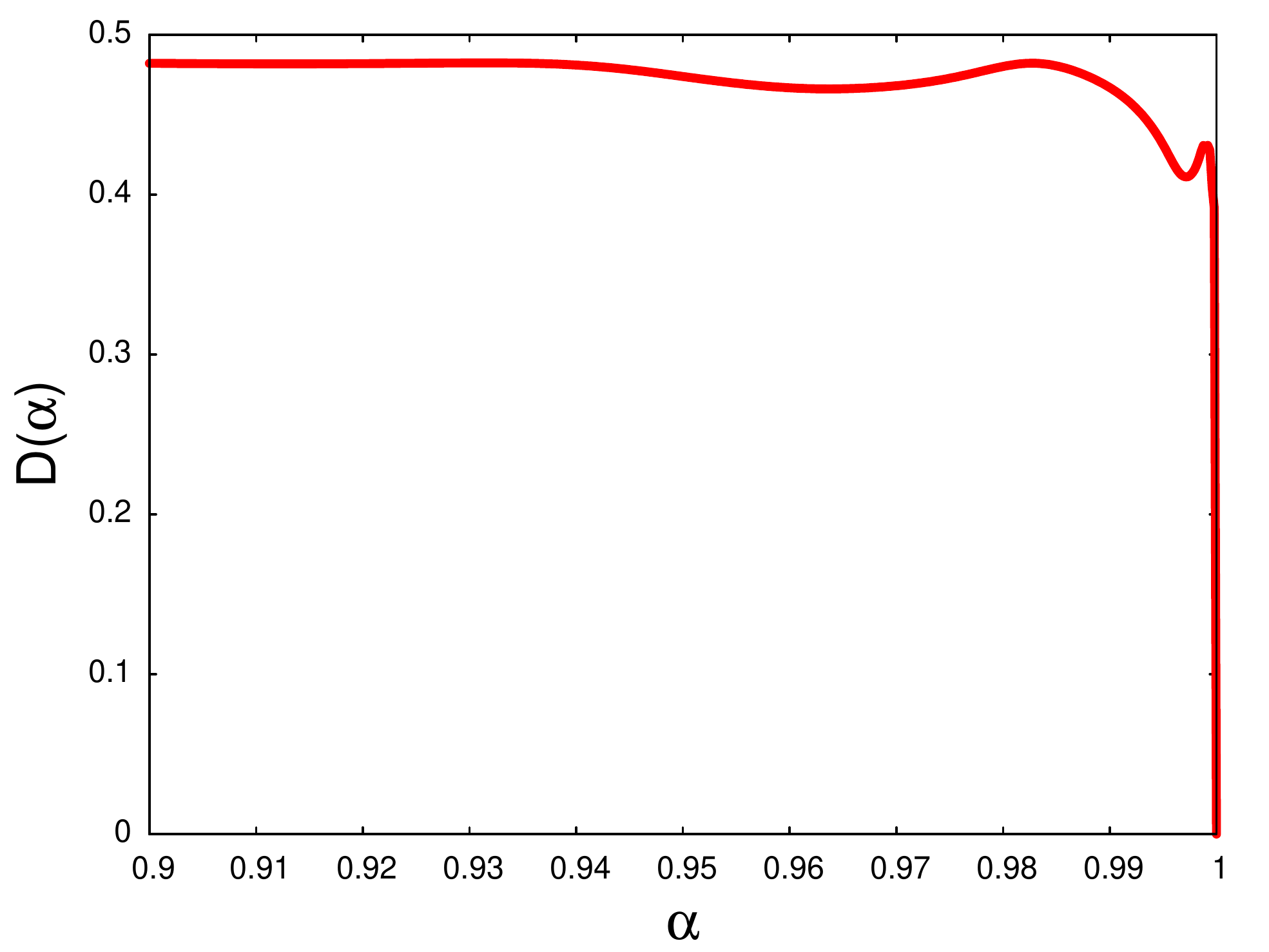}}\ 
\caption{Plots of $D(\alpha)$ for representatives of each of the classes defined in Table \ref{tab:ex1-classes}.}\label{pic:ex1}
\end{figure}

\begin{table}
\centering
\begin{tabular}{r|l}
\hline
{\bf Class I} &
0, 8, 12, 32, 40, 64, 68, 72, 76, 77, 93, 96, 128, 132,\\ &  136, 140, 160, 168, 192, 196, 200, 205--207, 220, 221, \\ & 224, 232, 233, 235--239, 249--255 \\ \hline
{\bf Class II} &
1--5, 7, 10, 13, 15--17, 19, 21, 23, 24, 29, 31, 34, 36,\\ & 42, 44, 48, 50,  51, 55, 56, 63, 66, 69, 71, 79, 80, 85, \\ & 87, 92, 95, 100, 104, 108, 112, 119, 127, 130, 138,\\ & 141, 144, 152, 162, 164, 170--172, 174--176, 178, 179, \\ & 186--191, 194, 197, 201--203, 208, 216--219, 222, 223, \\ & 226, 228, 230, 231, 234, 240--248 \\ \hline
{\bf Class IIIa} &
6, 9, 18, 20, 22, 25--28, 30, 33, 35, 37--39, 41, 45, 46, \\ & 49, 52--54, 57--62, 65, 67, 70, 73--75, 78, 82, 83, 86,\\ &  88--91, 94, 97--99, 101--103, 105--107, 109--111, 114--116, \\ &118, 120--126, 129, 131, 133--135, 137, 139, 145--151, \\ & 153--159, 161, 163, 165--167, 169, 173, 177, 180--185, \\ &193, 195, 198, 199,204, 209--211, 214, 215, 225, 227, 229 \\ \hline
{\bf Class IIIb} &
11, 14, 43, 47, 81, 84, 113, 117, 142, 143, 212, 213 \\ \hline
\end{tabular}
\caption{Classes of $\alpha$--Asynchronous ECAs.}
\label{tab:ex1-classes}
\end{table}

The interpretation of the classes is as follows. Rules belonging to Class~I are resistant to $\alpha$--asynchronicity, which means that introducing the asynchronicity aspect does not affect the behavior, or that its impact is negligible. Rules belonging to Class II are affected by the asynchronicity, but the impact on the final behavior is proportional to the synchrony rate $\alpha$, {\it i.e.}\ the behavior steadily approaches the deterministic case as $\alpha$ approaches one. Rules belonging to Classes~IIIa and IIIb are sensitive to $\alpha$--asynchronicity, meaning that the behavior of the system changes drastically as soon as $\alpha$ is smaller than 1. Class~IIIb can be distinguished from Class IIIa, by the noisy behavior of $D(\alpha)$ for $\alpha$ close to one. The cause of the differences is not yet uncovered.

\hl{Classes IIIa and IIIb illustrate the case where the decomposition as a stochastic mixture does not suffice to understand the dynamics of the SCA. Indeed, for large $\alpha$, the behavior of the SCA and that of the deterministic CA with the highest probability of application is quite different.}

\subsection{Stochastic density classifiers}
We consider the SCA defined by the local rule $\mathbf{C}_3$ introduced in \cite{fates:inria-00608485,nazim-automata2013}. This rule is defined by the pLUT shown in Table \ref{tab:plut-nazim}. Note that this pLUT uses the standard notation for its state set $\mathcal{S}=\{0,1\}$, and its entries constitute the probabilities of reaching state $1$. It was shown both analytically and experimentally that this rule is a stochastic solution for the classical density classification problem (DCP) \cite{fates:inria-00608485} with arbitrary precision, {\it i.e.}\ by varying the parameter $\eta$, we can achieve any probability $p<1$ of successful classification. The DCP was introduced in \cite{gacs1978one,packard1988adaptation}. 

\renewcommand{\tabcolsep}{.1cm}
\begin{table}[ht]
\centering
\begin{tabular}{>{$}c<{$}|>{$}c<{$}|>{$}c<{$}|>{$}c<{$}|>{$}c<{$}|>{$}c<{$}|>{$}c<{$}|>{$}c<{$}}
\hline
(1,1,1) & (1,1,0) & (1,0,1) & (1,0,0) & (0,1,1) & (0,1,0) & (0,0,1) & (0,0,0) \\
\hline
1 & \eta & 1 & 1-\eta & 1 & 0 & 0 & 0 \\
\hline
\end{tabular}
\caption{pLUT of local rule $\mathbf{C}_3$ \cite{fates:inria-00608485,nazim-automata2013}.}\label{tab:plut-nazim}
\end{table}

Using the CCA representation, we can visualize the average behavior of~$\mathbf{C}_3$. In Figure \ref{pic:c3}, the first 80 rows of the averaged space-time diagram (time goes from top to bottom) obtained using the CCA representation of $\mathbf{C}_3$ with $\eta=0.1$ are shown, together with three typical space-time diagrams obtained by direct evaluation of the SCA rule $\mathbf{C}_3$. All images were obtained using the same initial condition involving 29 cells out of which 16 were black (state 1). As can be inferred from Figure \ref{pic:c3}, two samples lead to behavior that is similar to the one displayed by the average space-time diagram, while the one depicted in Figure \ref{fig:fail} behaves differently, and leads to a wrong classification. Since the probability of obtaining a correct classification for the initial condition at stake was estimated to be 90.1\% on the basis of 10000 repetitions, the impact of the erroneous cases on the averaged image is hardly visible.

\begin{figure}
\centering
\subfloat[Averaged]{\includegraphics[width=0.19\textwidth,fbox]{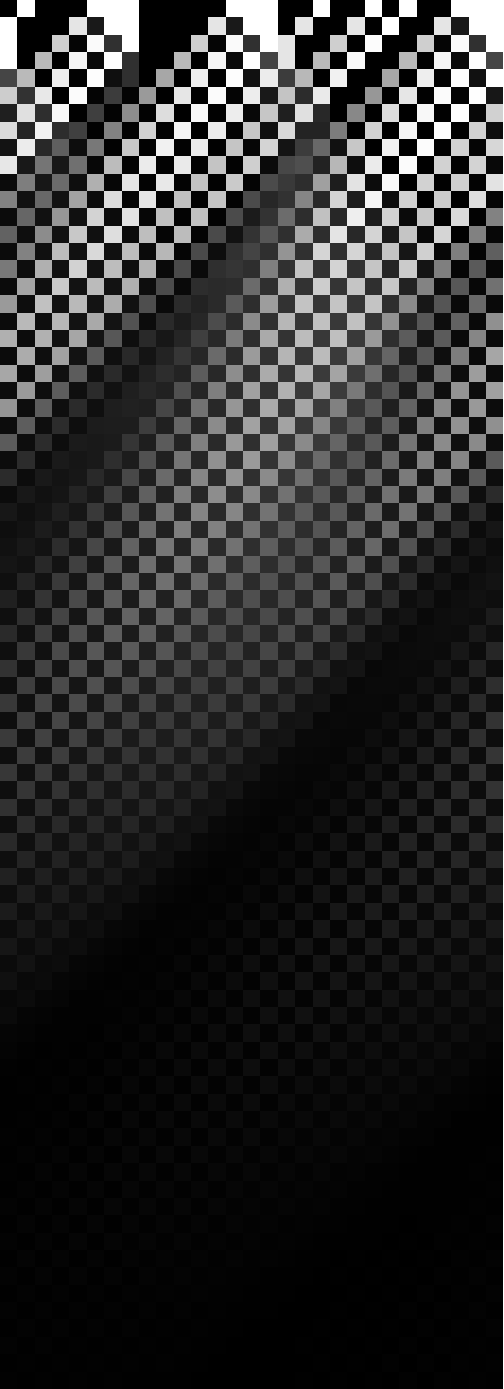}}\quad
\subfloat[Sample \#1]{\includegraphics[width=0.19\textwidth,fbox]{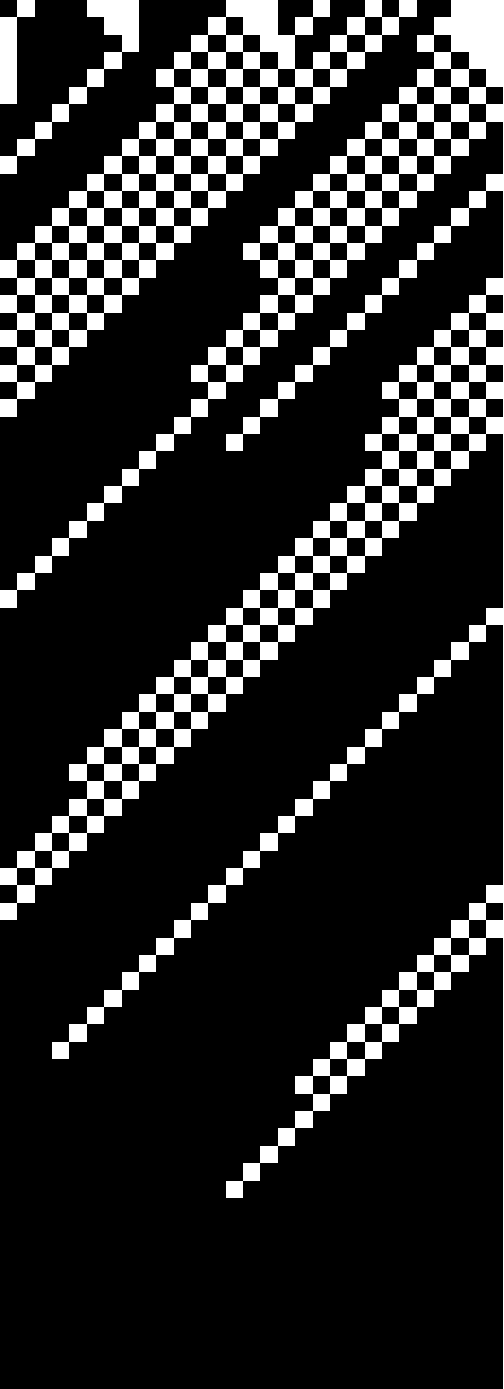}}\quad
\subfloat[Sample \#2]{\includegraphics[width=0.19\textwidth,fbox]{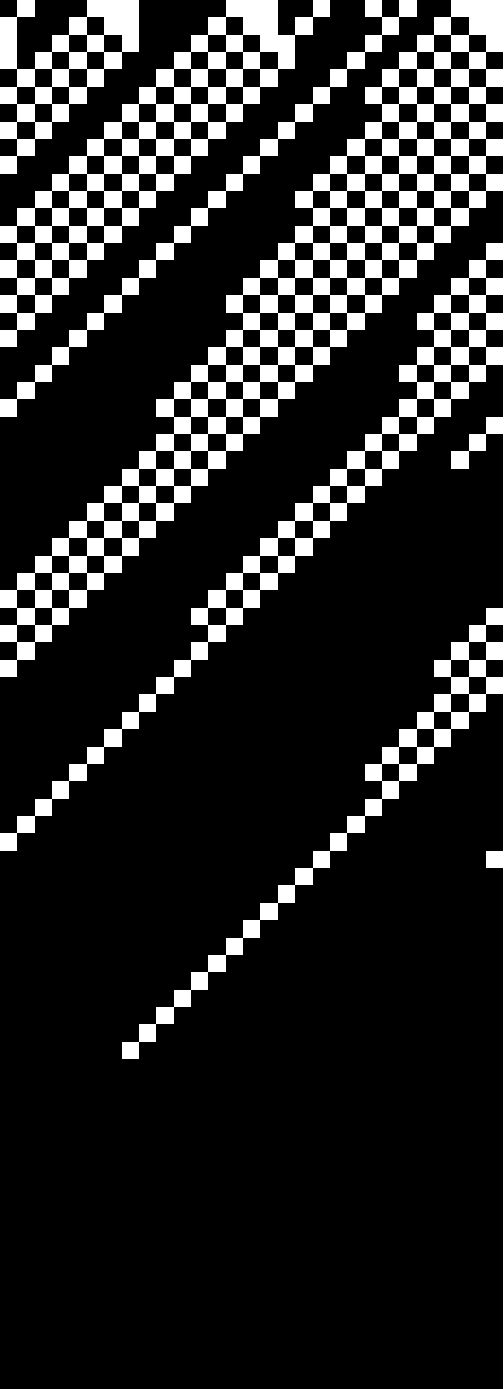}}\quad
\subfloat[Failing sample]{\label{fig:fail}\includegraphics[width=0.19\textwidth,fbox]{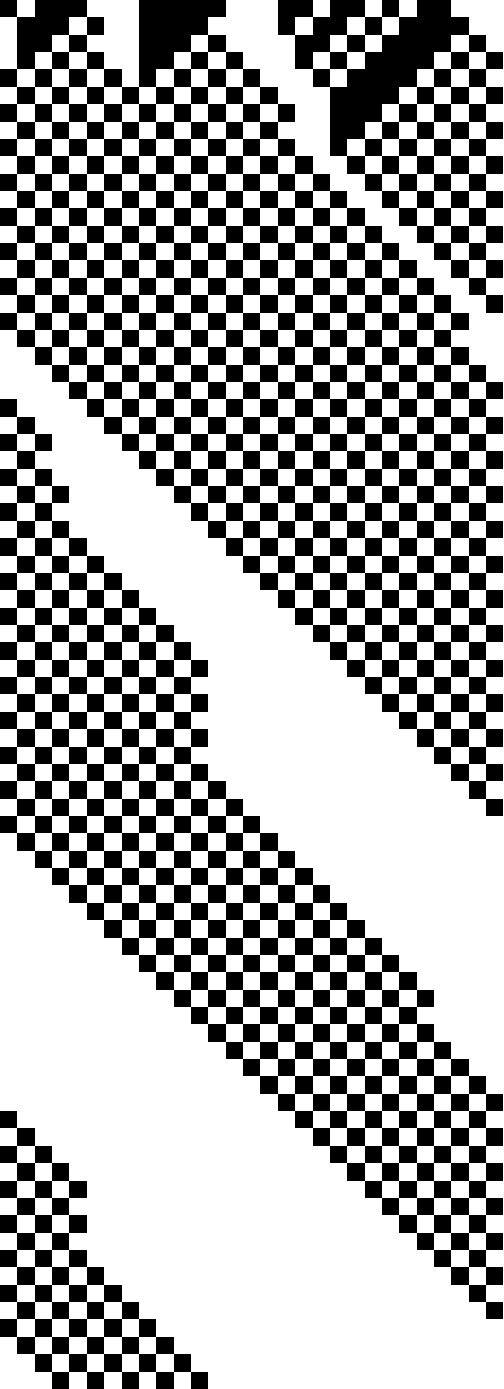}}\quad
\caption{Space time diagrams of $\mathbf{C}_3$ with $\eta = 0.1$.}\label{pic:c3}
\end{figure}

One of the most interesting questions within the analysis of DCP solutions relates to the expected time of convergence towards the outcome of the classification and how this relates to the number of cells. Figure \ref{pic:c3-conv-time} depicts the average convergence time calculated over an ensemble of 5000 random initial conditions for different numbers of cells and values of $\eta$, both for the CCA and the SCA representation. The initial conditions were generated as follows. For every initial condition independently, a probability $p$ was selected randomly, and then the initial states were selected with $p$ being the probability of selecting state 0 at each cell independently. Such a selection procedure assures that each of the possible densities has the same probability in the ensemble of initial conditions, and thus we can evaluate the classifier across a diverse set of densities.

In the case of the CCA representation, we cannot expect reaching a truly homogeneous, global state. Therefore we evolved it until the maximum, absolute difference between the states of any two cells was lower than 0.001. Then, we verified whether the average of the states was closer to 1 or 0. In the case of the SCA, for each of the 5000 initial conditions, 100 simulations were performed, and the results were averaged. 

\begin{figure}
\centering
\subfloat[CCA]{\includegraphics[width=0.45\textwidth]{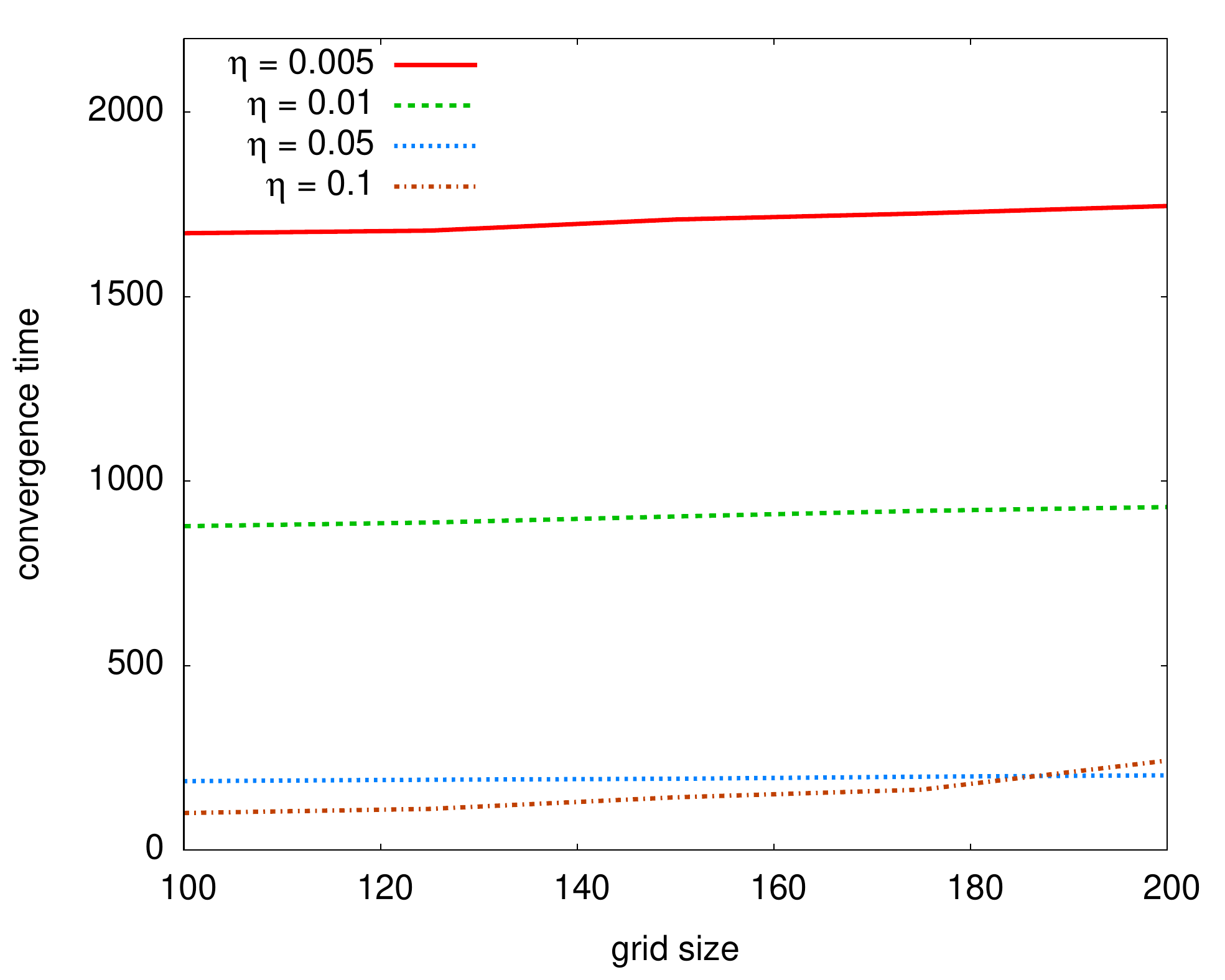}}\quad
\subfloat[SCA]{\includegraphics[width=0.45\textwidth]{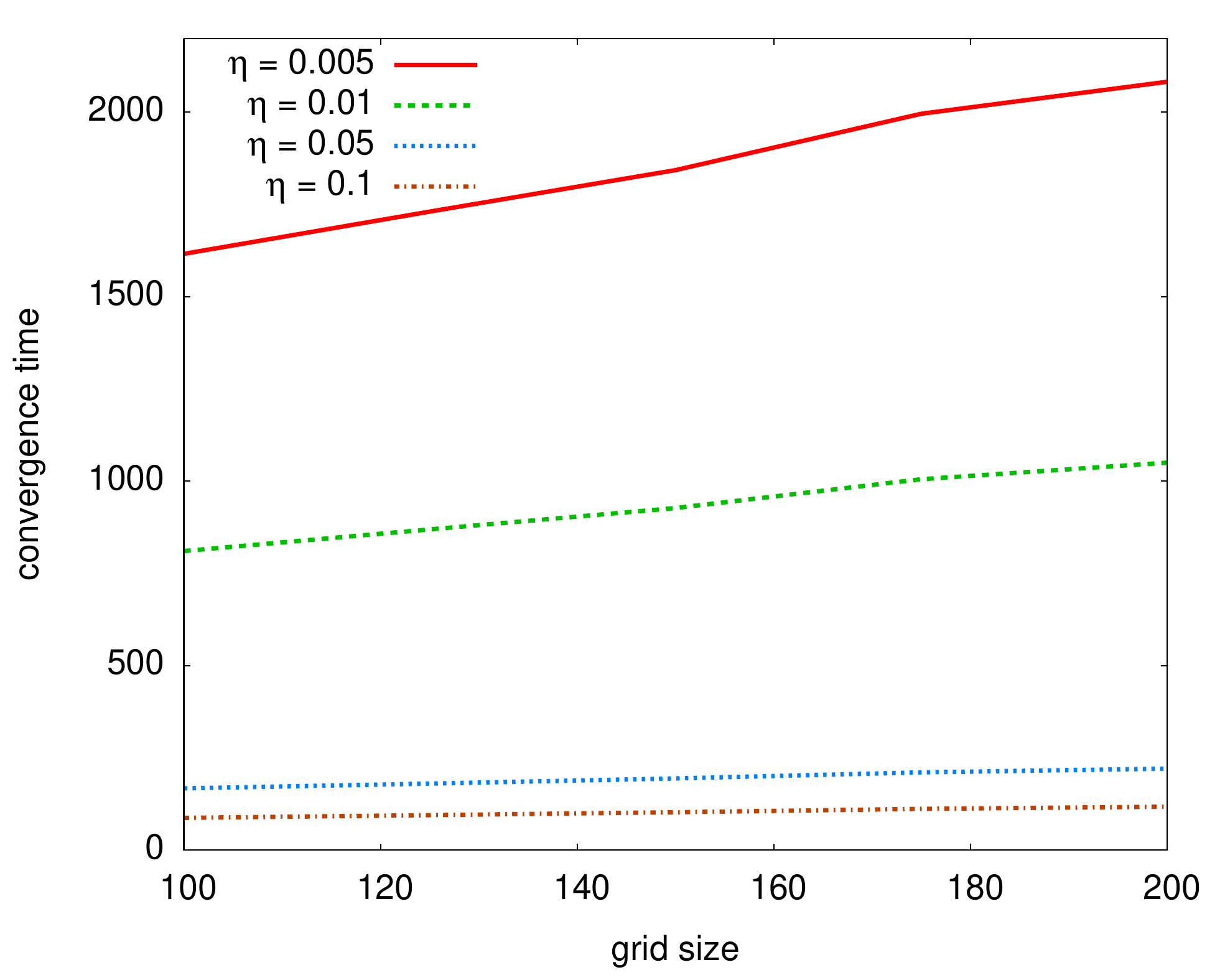}}
\caption{Expected convergence time of rule $\mathbf{C}_3$ for different values of $\eta$.}\label{pic:c3-conv-time}
\end{figure}

The results for small grid sizes (smaller than 200) are similar for both representations and agree with the findings presented in \cite{fates:inria-00608485,nazim-automata2013}. Yet, the computing time when using the CCA representation is significantly lower than in the case of SCAs. This indicates that the CCA representation of SCAs proposed in this paper might enable a substantial reduction of the required computing time. Moreover, using the CCA representation, we are able to get more insight into the behavior of the dynamical system by plotting the evolution of the density over time. Figure \ref{pic:c3-conv-time2} shows the results of such an experiment. For the sake of clarity, these results are based on an ensemble of 400 initial conditions (69 cells) out of which 200 had a density greater than 0.5 (green), while the other 200 had a density smaller than 0.5 (red). The ensemble of initial configurations was the same for all values of $\eta$. Note that the range of the horizontal axis differs across plots. From Figure \ref{pic:c3-conv-time2}, we can infer that the time and quality of classification increases with decreasing $\eta$. Due to the computational costs involved in evolving and averaging SCAs directly, drawing and analyzing such plots has become possible only due to the introduction of CCAs.

\begin{figure}
\centering
\subfloat[$\eta=0.1$]{\includegraphics[width=0.48\textwidth]{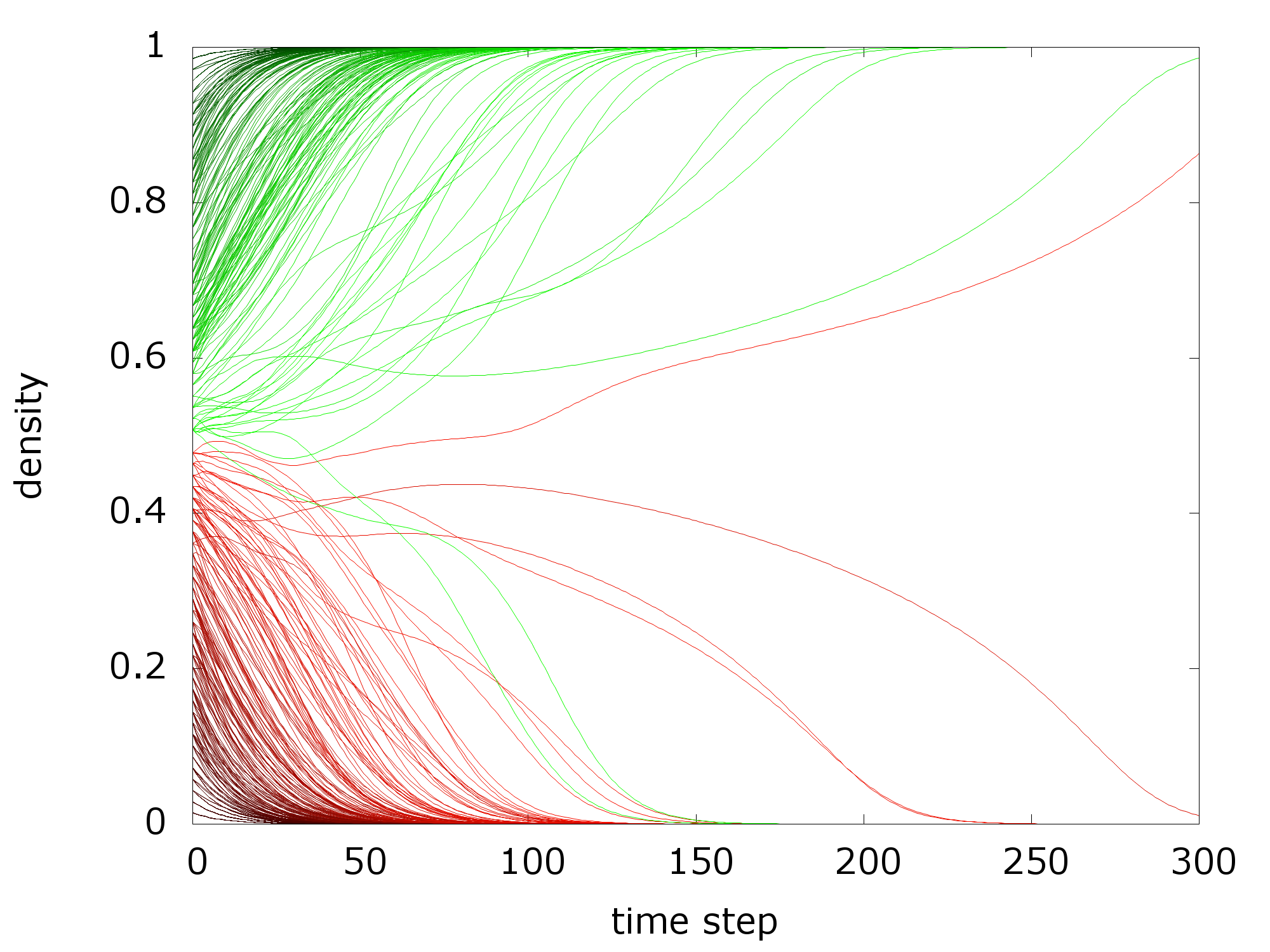}}
\subfloat[$\eta=0.05$]{\includegraphics[width=0.48\textwidth]{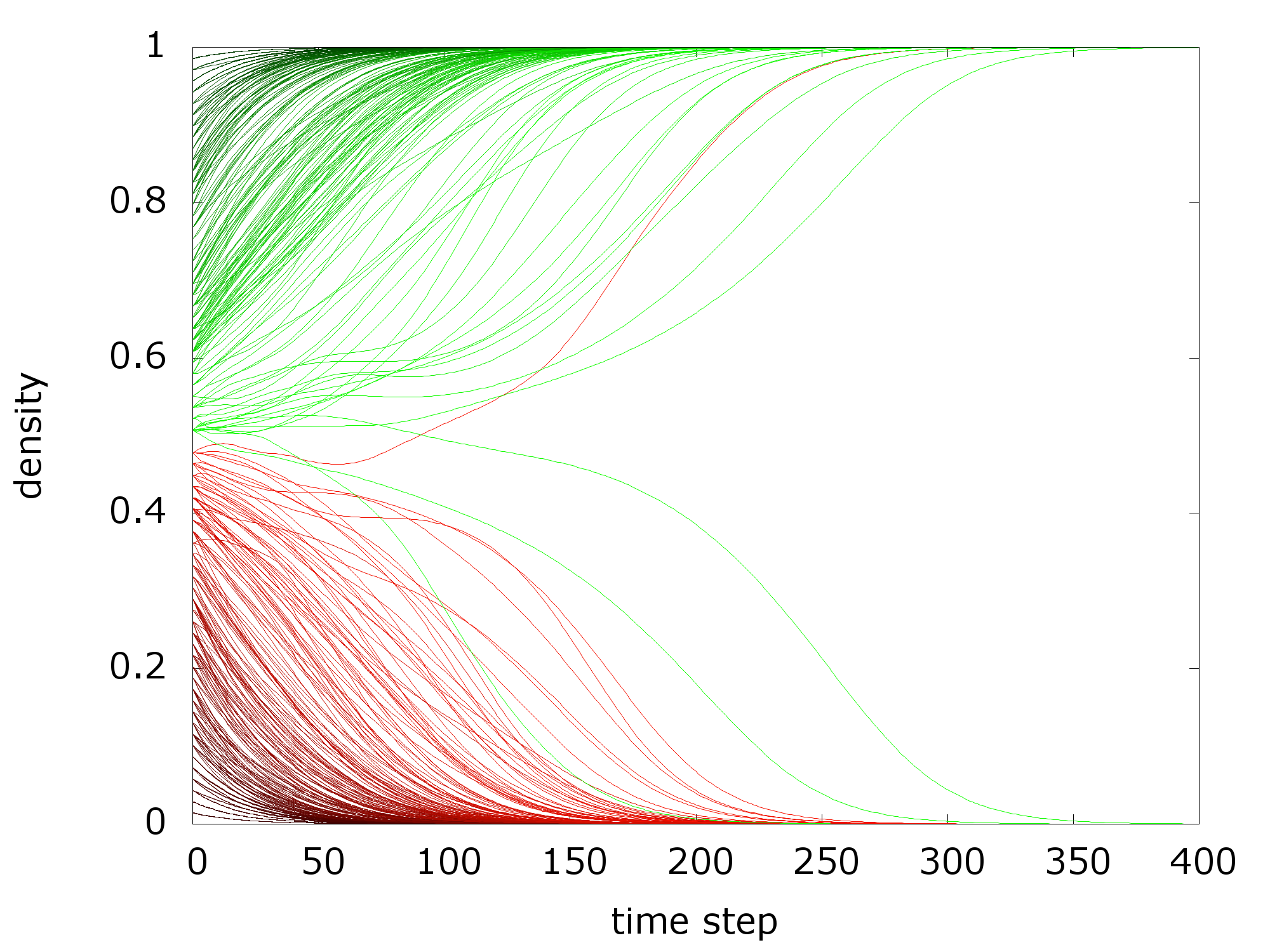}}\\
\subfloat[$\eta=0.01$]{\includegraphics[width=0.48\textwidth]{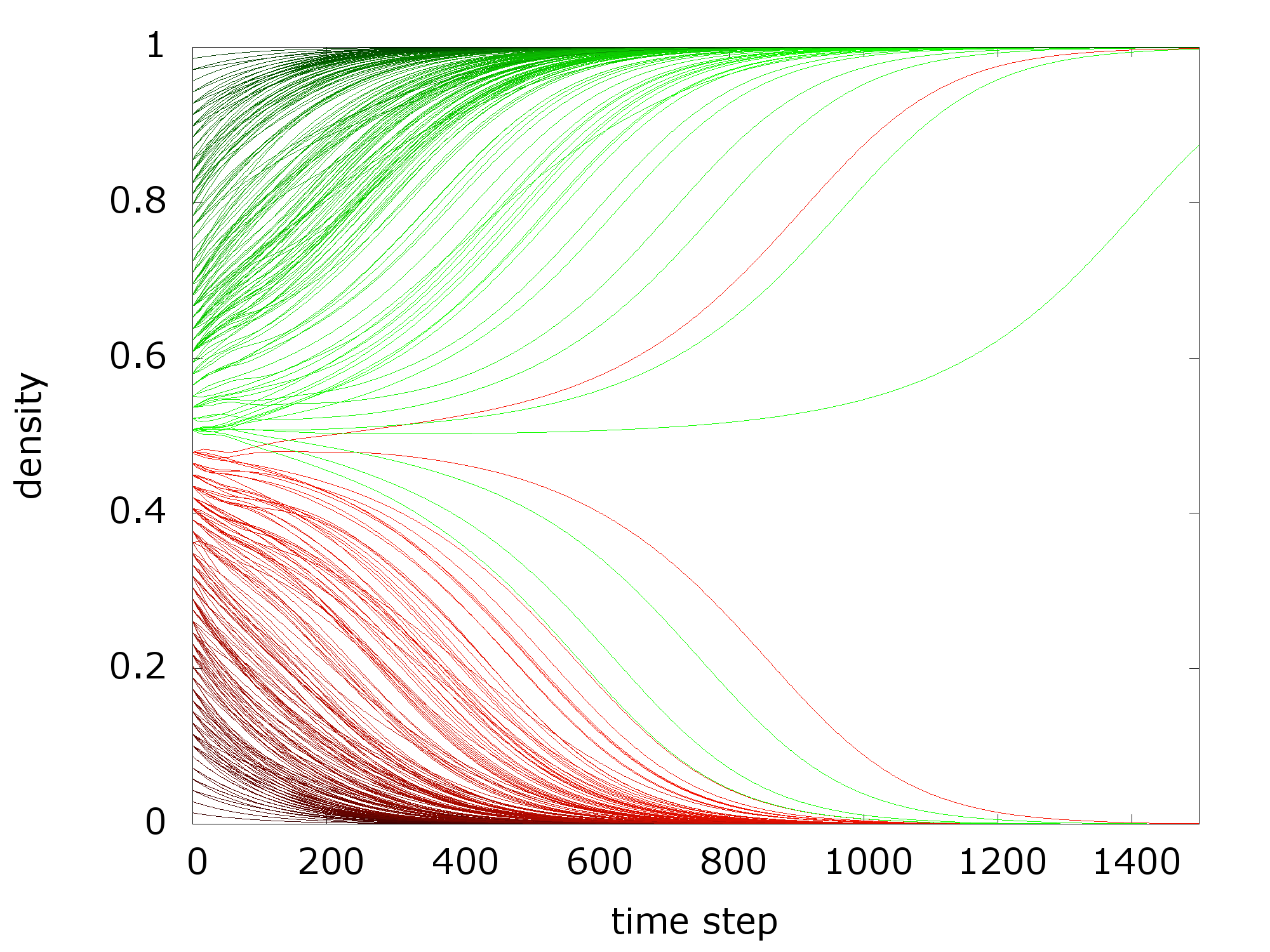}}
\subfloat[$\eta=0.005$]{\includegraphics[width=0.48\textwidth]{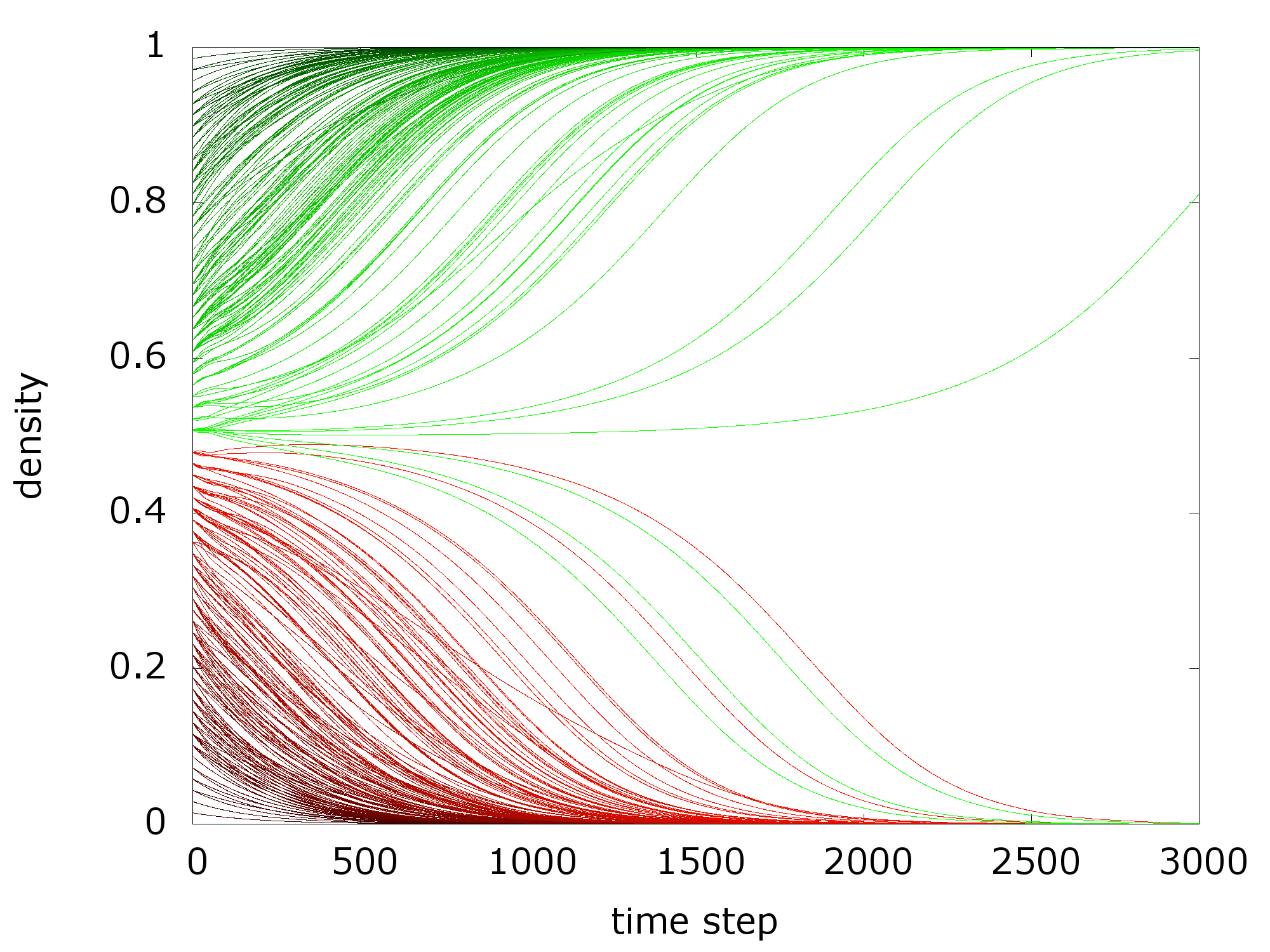}} \\
\subfloat[$\eta=0.001$]{\includegraphics[width=0.48\textwidth]{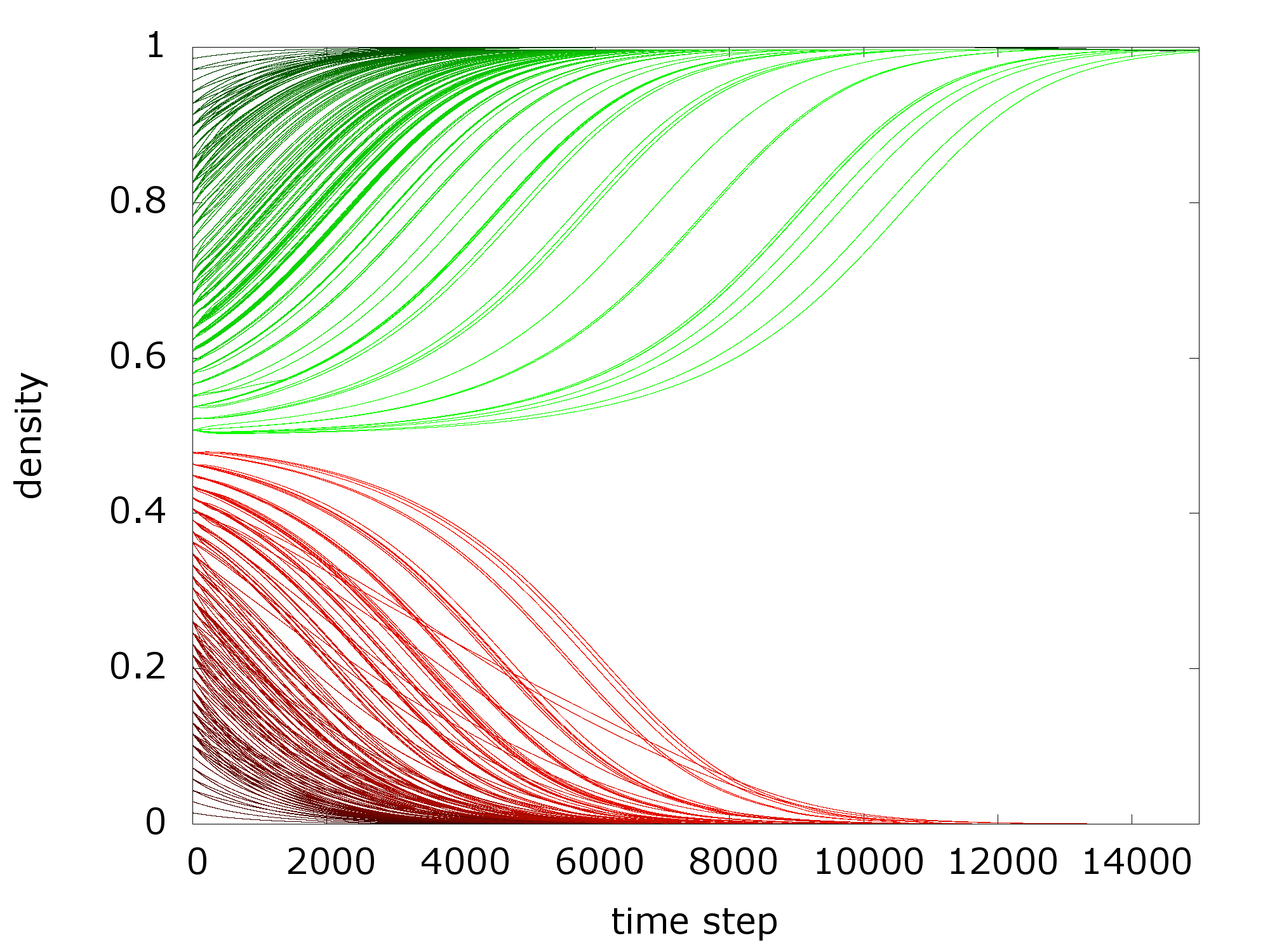}}
\caption{Density evolution over time for the CCA representation of SCA $\mathbf{C}_3$ for 200 initial conditions with density greater than 0.5 (green) and 200 other initial conditions with density smaller than 0.5 (red) and different values $\eta$.}\label{pic:c3-conv-time2}
\end{figure}

\subsection{Totalistic SCAs}
\hl{In this section, we analyze a specific class of totalistic SCAs. Although most likely being of no practical use, it holds many similarities with more complex CAs and SCAs that are often resorted to in modelling. Totalistic CAs are generally known for their ability to mimic real-world phenomena \cite{RevModPhys.55.601,wolfram-class,total}. Given their importance and straightforward formulation, we have chosen this class to serve as an example for our analysis. The focus is on the illustration of the tools introduced in this paper, rather than to unveil new properties of some complex CA models.}

We consider the class of 1D, binary SCAs with a unit neighborhood radius, that satisfy the following conditions:
\begin{align*}
\mathbb{P}\big(s(c_i,t+1) = e_2 \mid &\ s(c_{i-1}, c_i, c_{i+1}, t) = (e_2,e_1,e_1)\big) = \\
\mathbb{P}\big(s(c_i,t+1) = e_2 \mid &\ s(c_{i-1}, c_i, c_{i+1}, t) = (e_1,e_2,e_1)\big) = \\
\mathbb{P}\big(s(c_i,t+1) = e_2 \mid &\ s(c_{i-1}, c_i, c_{i+1}, t) = (e_1,e_1,e_2)\big) = p_1,
\end{align*}
\begin{align*}
\mathbb{P}\big(s(c_i,t+1) = e_2 \mid &\ s(c_{i-1}, c_i, c_{i+1}, t) = (e_1,e_2,e_2)\big) = \\
\mathbb{P}\big(s(c_i,t+1) = e_2 \mid &\ s(c_{i-1}, c_i, c_{i+1}, t) = (e_2,e_1,e_2)\big) = \\
\mathbb{P}\big(s(c_i,t+1) = e_2 \mid &\ s(c_{i-1}, c_i, c_{i+1}, t) = (e_2,e_2,e_1)\big) = p_2,
\end{align*}
where $e_1$ and $e_2$ are the base vectors of the Euclidean space $\mathbb{R}^2$. 
Such SCAs are commonly referred to as totalistic SCAs. We consider the subclass of such SCAs satisfying:
\begin{align*}\mathbb{P}\big(s(c_i,t+1) = e_2 \mid s(c_{i-1}, c_i, c_{i+1}, t) = (e_2,e_2,e_2)\big) = 1,\\
\mathbb{P}\big(s(c_i,t+1) = e_1 \mid s(c_{i-1}, c_i, c_{i+1}, t) = (e_1,e_1,e_1)\big) = 1.
\end{align*}
In this paper, we will refer to this subclass as totalistic SCAs. By applying the decomposition algorithm given by Proposition \ref{prop:stoch-mix-decom}, we can prove the following fact, revealing their structure.

\begin{fact}
Any totalistic SCA can be written as a stochastic mixture of the following ECAs: 128, 150, 232, 254, where up to three of those rules are applied with non-zero probability. 
\end{fact}

We can determine the regions in the $(p_1,p_2)$--plane where a given deterministic CA has the highest probability of application, {\it i.e.} where it is dominant. These regions are given by: 
\begin{itemize}
\item ECA 128 is dominant, if $p_1 < 0.5$ and $p_2 < 0.5$,
\item ECA 150 is dominant, if $p_1 < 0.5$ and $p_2 > 0.5$,
\item ECA 232 is dominant, if $p_1 > 0.5$ and $p_2 < 0.5$,
\item ECA 254 is dominant, if $p_1 > 0.5$ and $p_2 > 0.5$.
\end{itemize}

Note that if $p_1=0.5$ or $p_2=0.5$, there is no unique, dominant rule, therefore we have omitted those cases in the description above.

Let $\alpha_{128}$, $\alpha_{150}$, $\alpha_{232}$, $\alpha_{254}$ denote the probabilities of applying ECAs 128, 150, 232, 254, respectively. Figure \ref{fig:regions} depicts the dependence of $\alpha_{150}$ and $\alpha_{232}$ on $p_1$ and $p_2$. We have omitted the remaining two images for $\alpha_{128}$ and $\alpha_{254}$, since they do not yield additional information.

\begin{figure}
\centering
\subfloat[$\alpha_{150}$]{\includegraphics[width=0.48\textwidth]{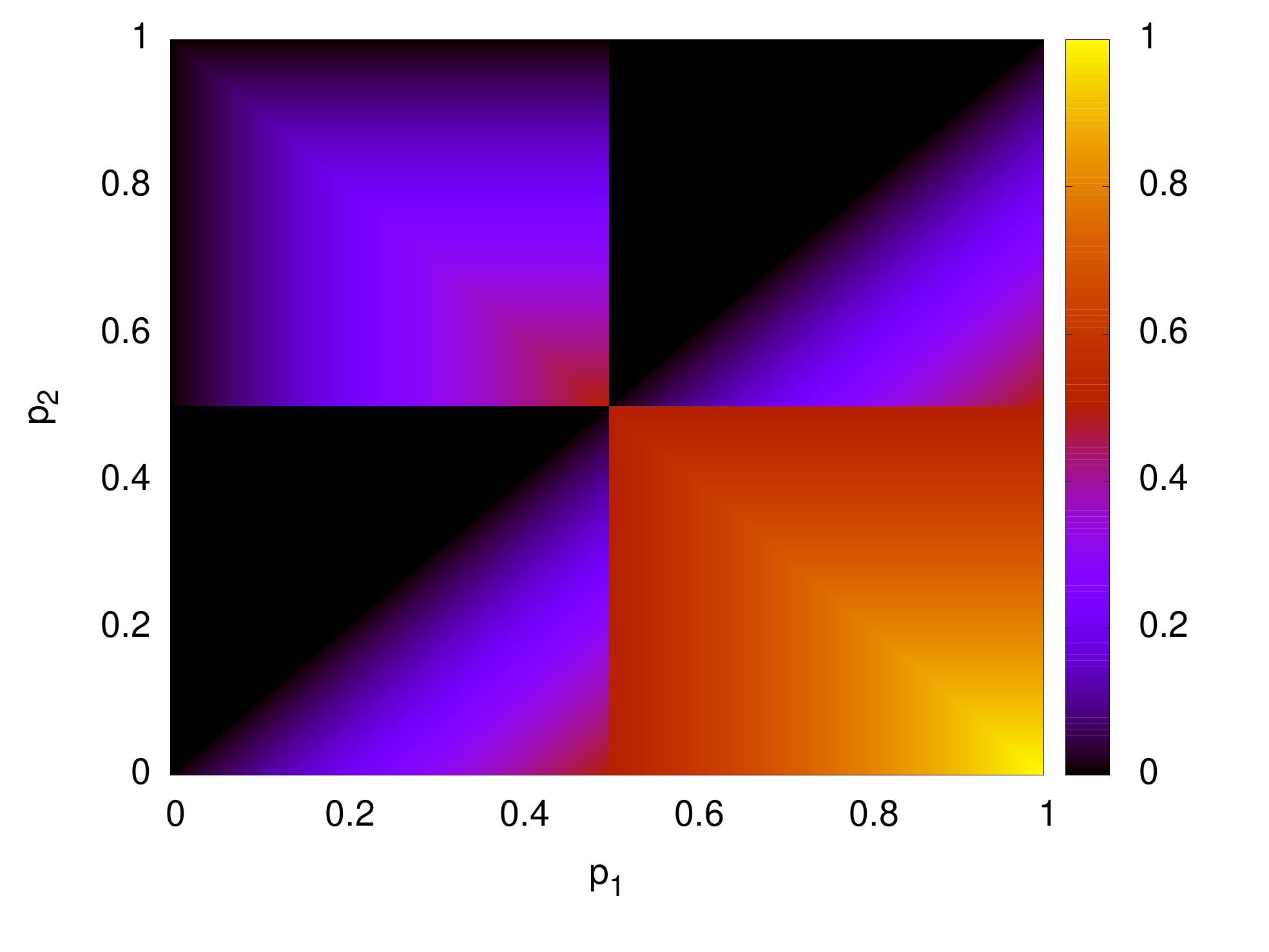}}
\subfloat[$\alpha_{232}$]{\includegraphics[width=0.48\textwidth]{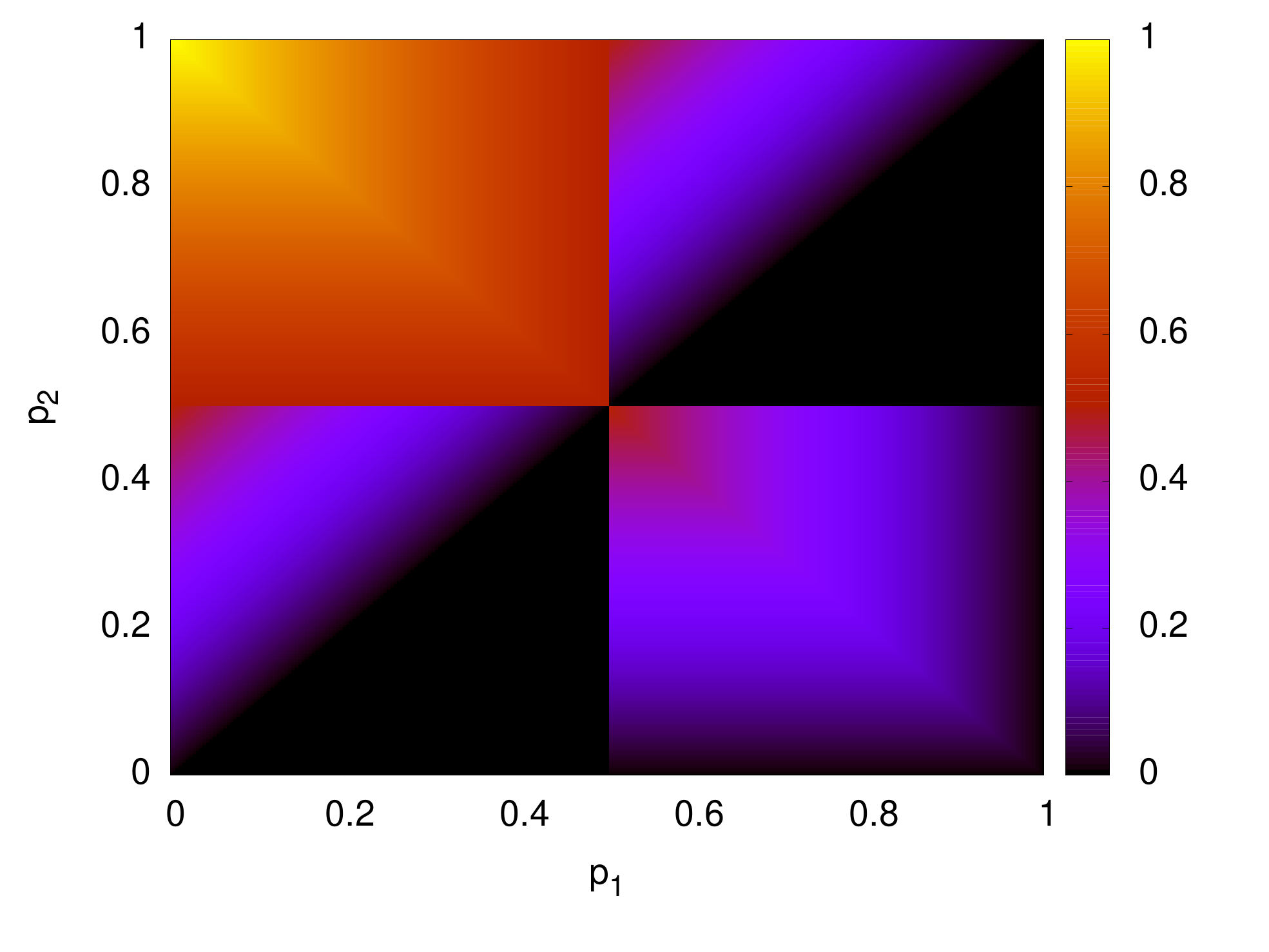}}
\caption{Dependence of the application probabilities: $\alpha_{150}$ of ECA 150 and $\alpha_{232}$ of ECA 232 in the totalistic SCA, on the values of the parameters $p_1$ and $p_2$.}\label{fig:regions}
\end{figure}

The dynamical characteristics of the four ECAs that compose the totalistic SCAs differ significantly. Rules 128 and 254 are simple -- they belong to Class~I according to Wolfram's  classification \cite{wolfram-class}, and their maximum Lyapunov exponents (MLEs) \cite{Bagnoli199234} equals $-\infty$, which means that any changes to their initial configuration do not influence their long-term behavior. Rule 232 belongs to  complexity Class II, and its MLE is positive, but very close to zero. In contrast, rule 150 is a Class III rule, and its MLE is the highest among all ECAs, which highlights the fact that this rule is sensitive to the smallest perturbation of its initial configuration. Since rules 150 and 232 are relatively more complex than rules 128 and 254, we might expect a distinct behavior of the stochastic mixture in regions where the former are dominant.

In order to verify this, we set up an experiment involving a random initial condition of $M=49$ cells which was evolved 100 times for $T=49$ time steps. Such a procedure, for the same initial condition, was repeated for multiple different choices of $(p_1,p_2)$ using a $101\times 101$ regular grid, resulting in a set of space-time diagrams. The set of such space-time diagrams is denoted as $\mathcal{I}(p_1,p_2)$. Let $\Delta(p_1, p_2) = \{ \dist(I,J) \mid I\neq J; I,J\in \mathcal{I}(p_1,p_2) \}$ denote the set of all pair-wise Hamming distances between space-time diagrams, and $\Delta_T(p_1, p_2) = \{ \dist(I[T],J[T]) \mid I\neq J; I,J\in \mathcal{I}(p_1,p_2) \}$ denote the set of all pair-wise Hamming distances between the final configurations in the space-time diagrams. The results for the minimum, mean and maximum of $\Delta(p_1,p_2)$ and $\Delta_T(p_1,p_2)$ are shown in Figure \ref{fig:diff}.

\begin{figure}
\centering
\subfloat[$\max \Delta(p_1,p_2)$]{\includegraphics[width=0.48\textwidth]{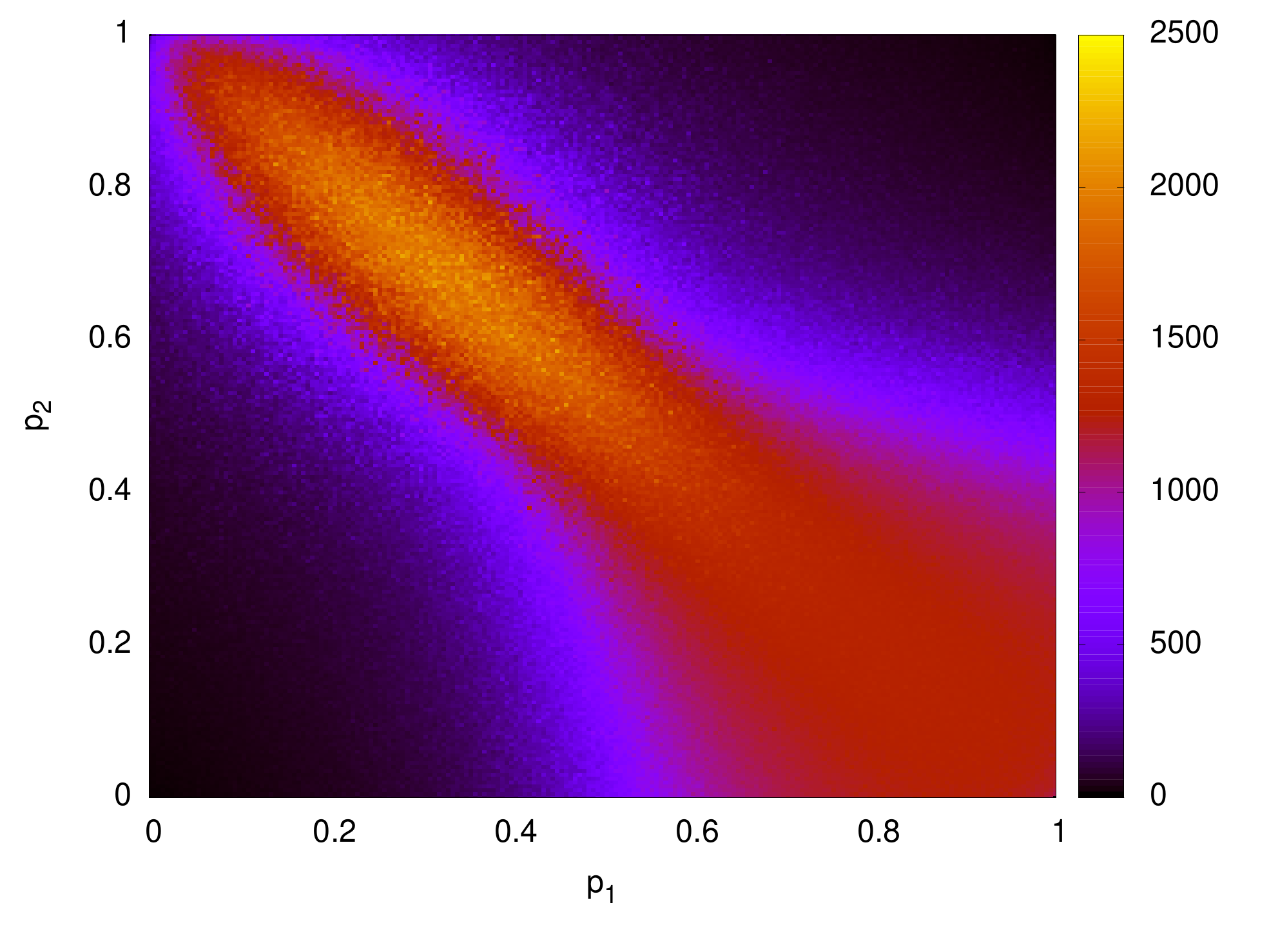}}
\subfloat[$\max \Delta_T(p_1,p_2)$]{\includegraphics[width=0.48\textwidth]{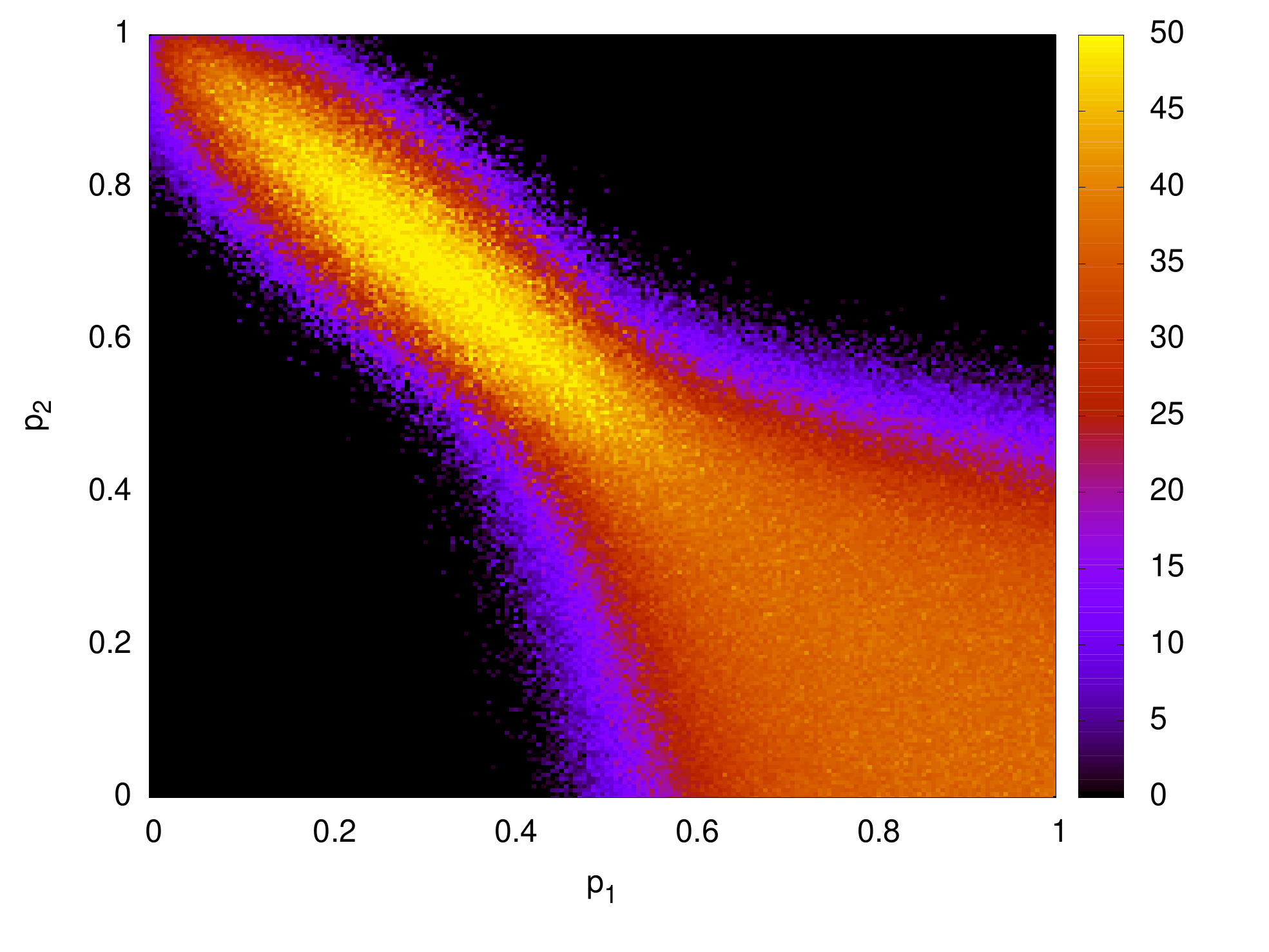}}\\
\subfloat[$\avg \Delta(p_1,p_2)$]{\includegraphics[width=0.48\textwidth]{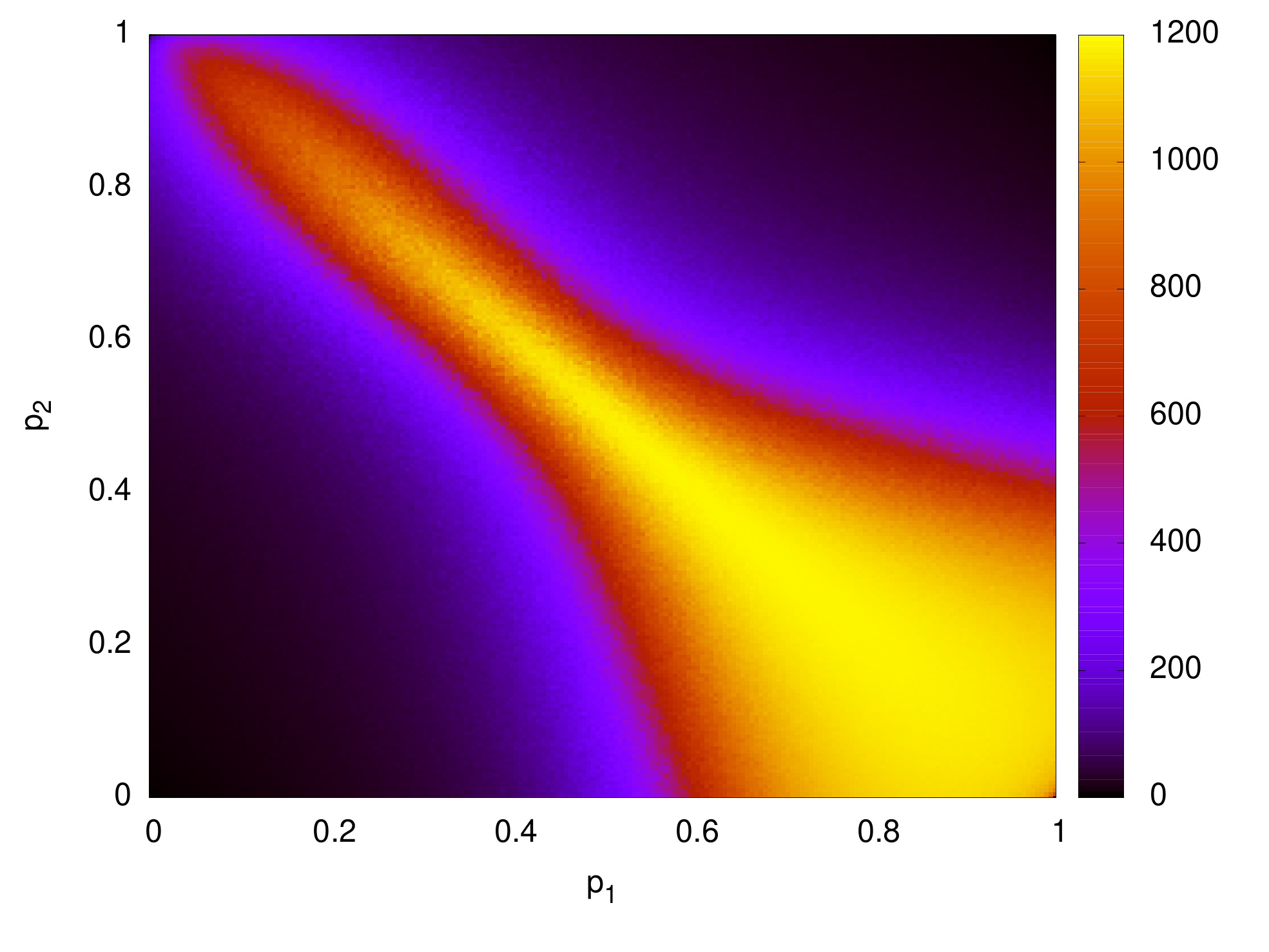}}
\subfloat[$\avg \Delta_T(p_1,p_2)$]{\includegraphics[width=0.48\textwidth]{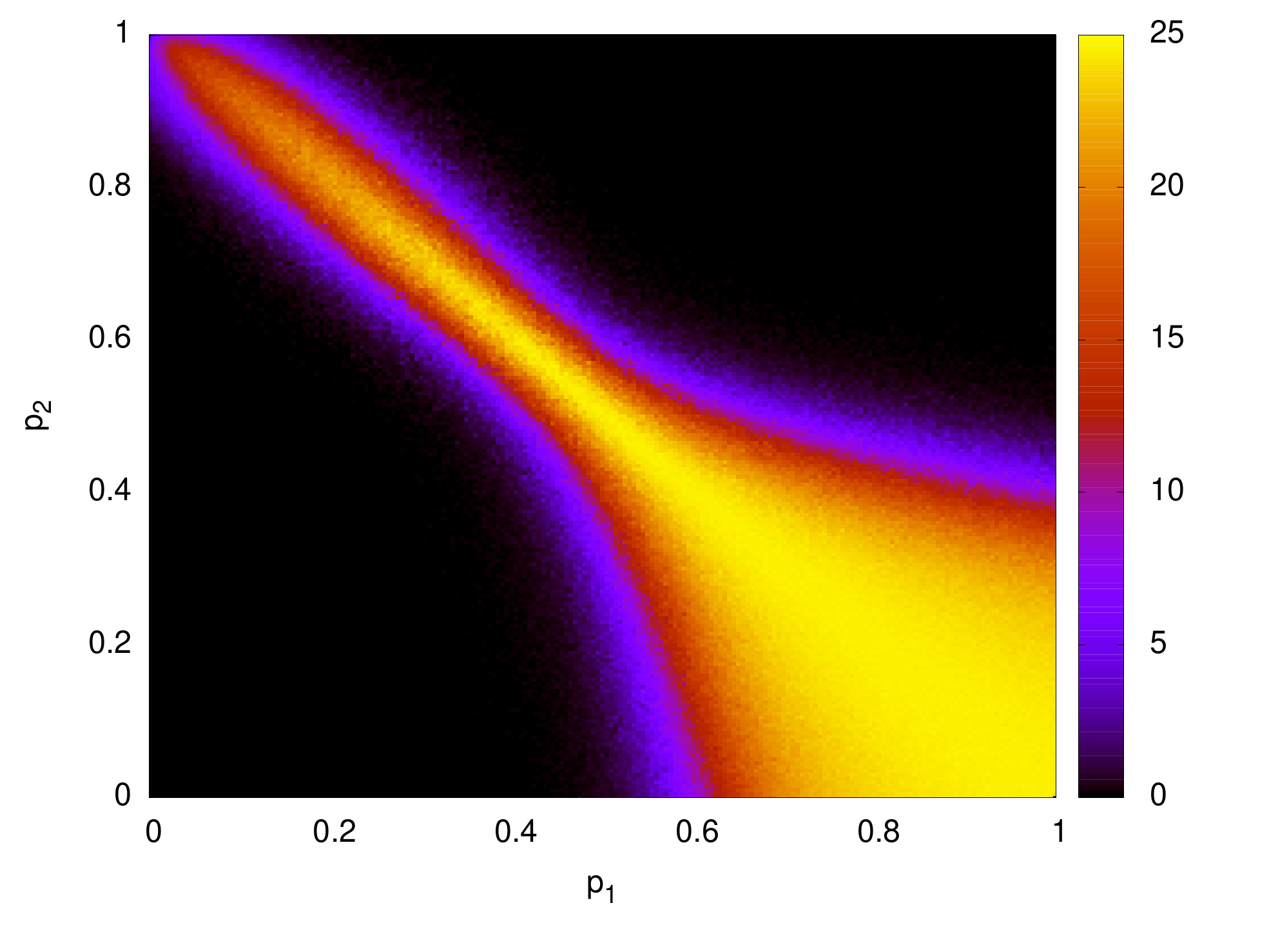}}\\
\subfloat[$\min \Delta(p_1,p_2)$]{\includegraphics[width=0.48\textwidth]{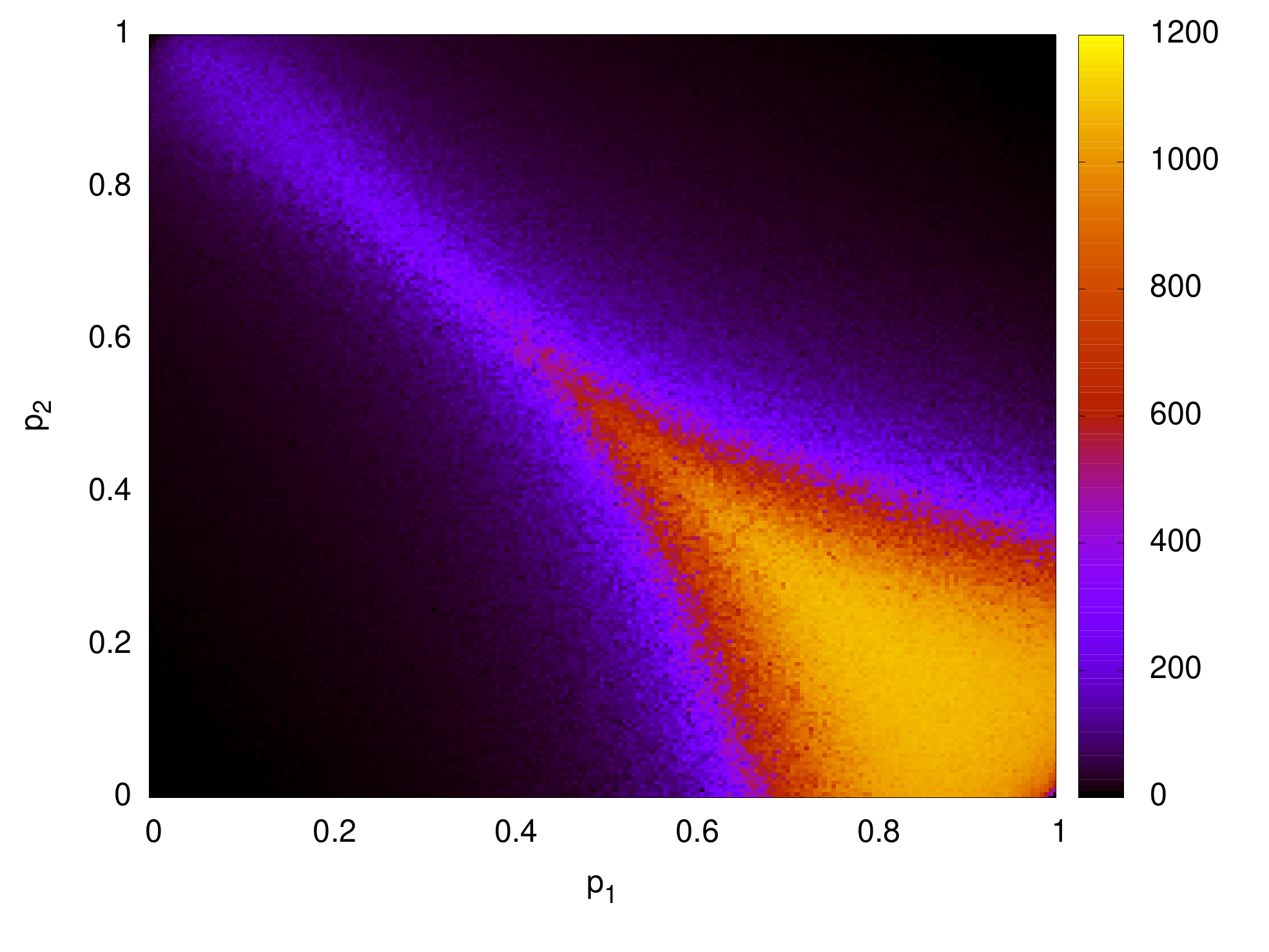}}
\subfloat[$\min \Delta_T(p_1,p_2)$]{\includegraphics[width=0.48\textwidth]{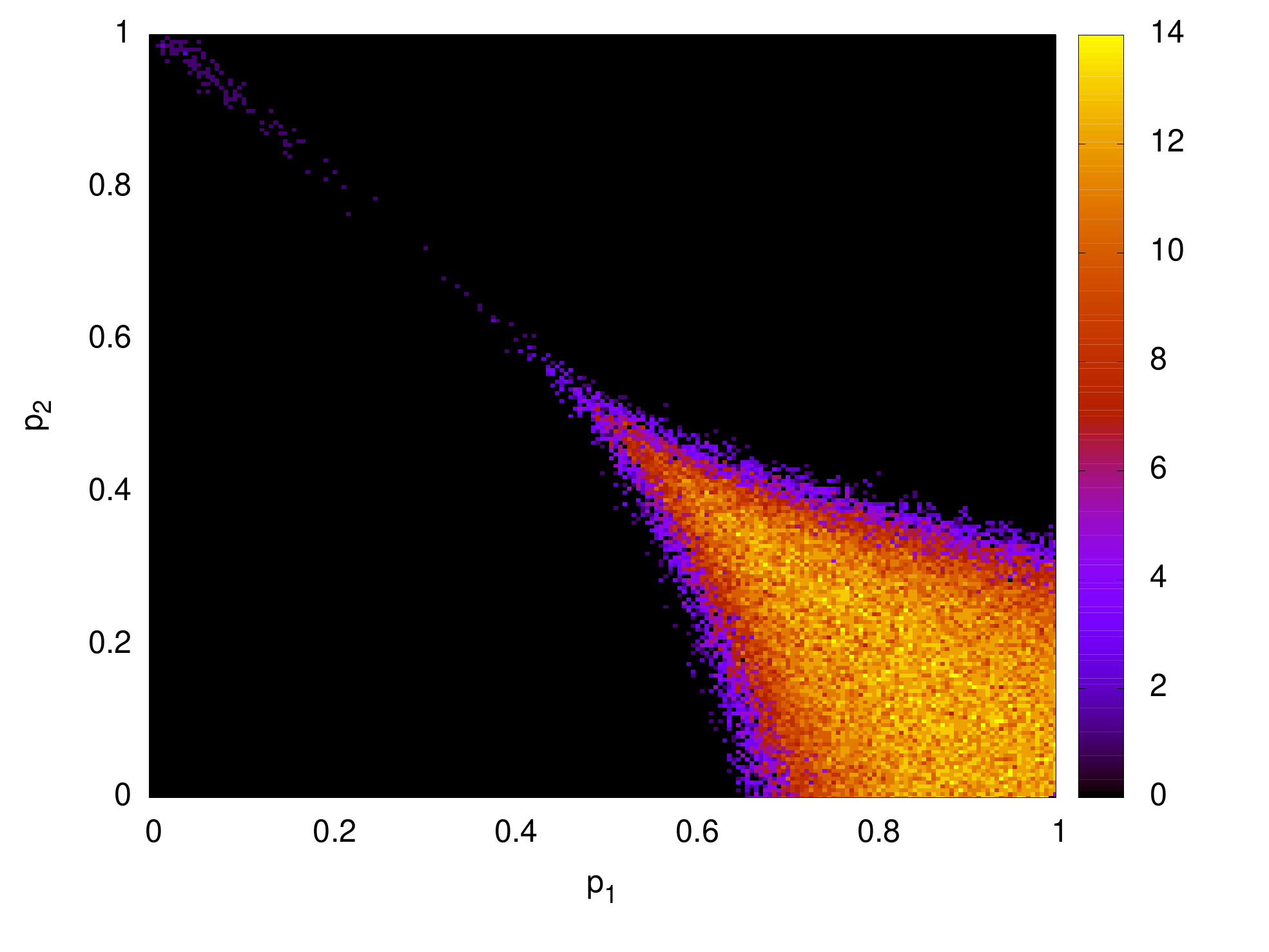}}
\caption{Impact of the probabilities $p_1$ and $p_2$ on the distances $\Delta(p_1,p_2)$ between the space-time diagrams (a, c, e) and the distances $\Delta_T(p_1,p_2)$ between the final configurations at $T=49$ (b, d, f), obtained obtained from the same initial condition.}\label{fig:diff}
\end{figure}

As can be seen from the charts, the distance is especially high in those regions where the application probability $\alpha_{150}$ is high, {\it i.e.}\ where $p_1>0.5$ and $p_2<0.5$, and near the line $p_2 = -p_1$ in the quarter $p_1<0.5$ and $p_2>0.5$. The first region is the one where ECA 150 is dominant, while ECA 232 is dominant in the second one, but with ECA 150 having a substantially high application probability as well. Although the Hamming distance is just a simple indication of complexity, we already see a strong influence of the ECA 150 component. 

\hl{The findings presented above suggest that analyzing the components of a stochastic mixture decomposition of an SCA unveils information on the dynamics of the SCA. In contrast to the $\alpha$-ACAs in Classes IIIa and IIIb ({\it cf.} Table~\ref{tab:ex1-classes}), here we encounter a class of SCA for which uncovering the components of the stochastic mixture with relatively high probability of application gives additional insight into the dynamical properties of the SCA.}

\section{Summary}
\label{sec:summary}
In this paper, two basic, yet important properties of SCAs were discussed. We have shown, both theoretically and practically, that SCAs can be effectively analyzed in the context of deterministic CAs and CCAs. The stochastic mixture representation of SCAs allows to understand the underlying dynamics, while the CCA representation allows to quickly uncover the average behavior of the system. Further research is undertaken to expand the application scope of the presented methods.   

\section*{Acknowledgments}
Witold Bołt is supported by the Foundation for Polish Science under International PhD Projects in Intelligent Computing. This project is financed by the European Union within the Innovative Economy Operational Program 2007--2013 and the European Regional Development Fund.

The authors gratefully acknowledge the financial support of the Research Foundation Flanders (FWO project 3G.0838.12.N).

\section*{Bibliography}
\bibliographystyle{elsarticle-num}
\bibliography{identify}
\end{document}